\documentclass[12pt]{iopart}

\usepackage{amsthm}

\newtheorem{lemma}{Lemma}[section]
\theoremstyle{definition}

\newtheorem{dfn}{Definition}

\usepackage{graphicx}
\usepackage{subfigure}
\usepackage{epsfig}

\newcommand{\prob}[1]{{\mathrm{Pr}}\left(#1\right)}
\newcommand{\be}{\begin{equation}}
\newcommand{\ee}{\end{equation}}
\newcommand{\ba}{\begin{array}}
\newcommand{\ea}{\end{array}}
\newcommand{\bea}{\begin{eqnarray}}
\newcommand{\eea}{\end{eqnarray}}
\newtheorem{prop}{Proposition}

\newcommand{\calL}{{\cal L }}

\newcommand{\calE}{{\cal E }}
\newcommand{\calZ}{{\cal Z }}
\newcommand{\calP}{{\cal P }}

\newcommand{\calS}{{\cal S }}
\newcommand{\calG}{{\cal G }}

\newcommand{\calC}{{\cal C }}

\newcommand{\la}{\langle}
\newcommand{\ra}{\rangle}

\begin{document}
% Title
\title[Constructions and noise threshold of topological subsystem codes]{Constructions and noise threshold of topological subsystem codes}

% Authors
\author{Martin Suchara$^1$, Sergey Bravyi$^2$ and Barbara Terhal$^2$}
\address{$^1$ Princeton University, Princeton, NJ 08544, USA}
\address{$^2$ IBM Watson Research Center, Yorktown Heights, NY 10598, USA}
\eads{\mailto{msuchara@princeton.edu}, \mailto{sbravyi@us.ibm.com}, \mailto{bterhal@gmail.com}}

% Abstract
\begin{abstract}
Topological subsystem codes proposed recently by Bombin are quantum error correcting codes defined on a two-dimensional grid of qubits that permit reliable quantum information storage with a constant error threshold. These codes require only the measurement of two-qubit nearest-neighbor operators for error correction. In this paper we demonstrate that topological subsystem codes (TSCs) can be viewed as generalizations of Kitaev's honeycomb model  to $3$-valent hypergraphs.  This new connection provides a systematic way of constructing TSCs and analyzing their properties. We also derive a necessary and sufficient condition under which a syndrome measurement in a subsystem code can be reduced to measurements of the gauge group generators. Furthermore, we propose and implement some candidate decoding algorithms for one particular TSC assuming perfect error correction. Our Monte Carlo simulations indicate that this code, which we call the five-squares code, has a threshold against depolarizing noise of at least 2\%.
\end{abstract}

% Keywords
\submitto{\JPA}

% Text
\maketitle

\section{Introduction}
\label{Introduction}

% What is active quantum error correction
Active quantum error correction~\cite{scheme_reducing_decoherence, simple_quantum_error, mixed_state_entanglement, theory_quantum_error} protects quantum information against noise by performing measurement to detect errors, and by correcting them. The goal of active error correction is to be able to store and manipulate quantum information arbitrarily long in the presence of noise provided that the noise level is below some critical threshold.

% Why stabilizer formalism and topological codes
The stabilizer formalism~\cite{class_quantum_error} provides a convenient framework for studying such error correcting codes. A special subclass of these codes are topological stabilizer codes~\cite{fault_tolerant_quantum, topological_quantum_memory}. Topological stabilizer codes, in particular the surface or toric code family \cite{fault_tolerant_quantum, BK:surface} offer practical advantages over other quantum error correcting codes. First, qubits can be laid out on a two-dimensional grid and error correction measurements can be performed in-place on the grid. Secondly, topological stabilizer and topological subsystem codes \cite{bombin:topsub} have noise thresholds which correspond to phase transitions in associated statistical mechanics models; these thresholds asymptote to a finite value in the (thermodynamic) limit of large block size, obviating the need for code concatenation. Thirdly, topological stabilizer and subsystem codes may allow for the implementation of various (but not all) gates such as the important CNOT gate by topological `code deformation' \cite{BMD:codedef, RH:cluster2D, bombin:codedef_topsub}. This topological implementation of logical gates does not require an additional overhead in the number or layout of qubits. Universality is obtained by supplementing the topological implementation of logical gates by magic state distillation \cite{BK:magicdist} for those gates which cannot be implemented topologically. An important advantage of this topological implementation of gates is that the tolerable noise rate for the code architecture will be largely determined by the quantum memory noise threshold. The quantum memory noise threshold for the surface code is among the highest of any codes that has been studied \cite{CDT:codestudy}.

% The depolarizing error model
In this paper, we will refer to and calculate thresholds for the depolarizing error model. In this model $X$, $Y$, and $Z$ errors occur on any given qubit with probability $p/3$, and no error occurs with probability $1-p$. The depolarizing probability is $p$. For noisy error correction on the surface code, the quantum memory noise threshold which has been achieved using explicit decoding is $0.75-1.1\%$ (see e.g. \cite{RHG:threshold, RH:cluster2D, wangetal:threshold, WFH:highthresh}.

In addition,  we will only consider noise-free error correction. This means that we assume that the syndrome measurements, --obtained for topological subsystem codes via the measurement of gauge group generators--, are noise free. With noise-free error correction the noise threshold of the surface code against depolarizing noise can be as high as $15-16\%$ \cite{DCP:topdec_prl, DCP:topdec_long, wangetal:threshold}. This number is close to the theoretical maximum in~\cite{WHP:threshold}.

% Motivation - the noise rate will decrease
Although the surface code is among the best code candidates developed so far, the actual experimental implementation of a surface code architecture is an extremely daunting task. Currently, Josephson junction qubits whose coherence times and capabilities have progressed rapidly over the past few years, seem to be well-suited for a 2D surface code architecture \cite{divincenzo:JJarchitecture}. With an anticipated elementary gate time of $\tau=O(10^{-8})$ seconds for Josephson junction qubits, a $T_2$ decoherence time of $O(10^{-5})$ seconds would lead to an effective noise rate of $\tau/T_2=O(10^{-3})$ {\em below} the noise threshold of the surface code.
% BMT the estimates above are what I hear is expected in the experiments these days

Beyond a sufficiently-low noise rate, a successful surface code architecture would have other technological requirements such the existence of a relatively fast measurement and the processing of syndrome data. Since the stabilizer group of the surface code is generated by four-body operators, syndrome measurement has to be performed by coupling a single ancilla qubit to four neighboring system qubits. In a Josephson junction architecture, this could for example mean that a single Josephson junction qubit would have to couple to four different resonators. The goal of this work is to ask whether one can simplify the physical implementation of error correction by employing topological subsystem codes \cite{bombin:topsub, subsystem_codes_spatially} in which 2-body (instead of 4-body) operators need to be measured for error correction. The formalism of subsystem codes and their relation to stabilizer codes were introduced and analyzed in~\cite{stabilizer_formalism_operator, operator_quantum_error}. Some subsystem codes such as the Bacon-Shor code family \cite{operator_quantum_error,AC:baconshor} are characterized by a stabilizer group which has non-local generators. These codes, although interesting, are not believed to have an asymptotically non-vanishing threshold, unlike the topological subsystem codes introduced by Bombin \cite{bombin:topsub}.

% Overview
After briefly introducing some stabilizer and subsystem code notation in Section~\ref{sec:background}, Section~\ref{sec:tsc} discusses ways of constructing subsystem codes with interesting properties. This method includes Bombin's topological subsystem codes, but can also give rise to new families of codes. Section~\ref{subs:properties} lists the desired code properties. Section~\ref{subs:honeycomb} sets the stage for our results by considering a `trivial' topological subsystem code based on Kitaev's honeycomb model (this code has no logical qubits). This code serves as a template for all code constructions introduced later. It also provides us with insights in how and under what assumptions the syndrome measurement can be implemented by measuring 2-body gauge group generators, see Section~\ref{subs:honeycomb} and Appendix~A. We will define extensions of the honeycomb model to $3$-valent hypergraphs and explain how these extensions can be used to derive topological subsystem codes, see Sections~\ref{subs:loops} and~\ref{subs:loops2gauge}. Our general construction is illustrated with two specific examples in Section~\ref{subs:examples} where we define the square-octagon code and the five-squares code. Finally, in Section~\ref{Decoding_algorithm} we describe a possible decoding algorithm for the five-squares code. We evaluate the performance of the decoding algorithm in Section~\ref{Experiments}.

\section{Background}
\label{sec:background}
Quantum error correcting codes \emph{encode} quantum information into some subspace $C$, --the code space--, of a larger Hilbert space. After the encoded information is subjected to noise, a \emph{syndrome measurement} is performed to diagnose the error that occurred. Finally, the syndrome information is used to return the code into a state that corresponds to the originally encoded information.

\subsection{Stabilizer Codes}
\label{subs:stabilizer}
We consider a system comprised of $n$ qubits. Let $X_j$, $Y_j$, and $Z_j$ represent the Pauli operators acting on the $j$-th qubit. Let $\mathcal{P}_n = \langle i\, I,X_1,Z_1,...,X_n,Z_n \rangle$ be the Pauli group on $n$ qubits. A stabilizer group $\mathcal{S}$  is an Abelian subgroup of $\mathcal{P}_n$ which does not contain $-I$. One can always choose independent generators $S_1, S_2,\ldots S_{n-k}\in \calS$ for some $k\le n$.

The codespace $C$ is the $2^{k}$ dimensional space stabilized by $\mathcal{S}$, i.e., the subspace of states $|\psi\rangle$ for which $S_j|\psi\rangle= |\psi\rangle$ holds for all $j=1,\ldots,n-k$. Such stabilizer code encodes $k$ logical qubits whose logical operators are denoted as $\{\bar{X}_j,\bar{Z}_j\}_{j=1,...,k}$. The centralizer  $\mathcal{C}(\mathcal{S})$ in ${\cal P}_n$ includes all Pauli operators in $\mathcal{P}_n$ that commute with every element of $\mathcal{S}$. For any stabilizer code one has the identity ${\cal C}(\mathcal{S})=\langle {\cal S}, \bar{X}_1,\bar{Z}_1,\ldots, \bar{X}_{k},\bar{Z}_{k}\rangle$. Hence the logical operators of the code are elements of  $\mathcal{C}(\mathcal{S}) \setminus \mathcal{S}$.

\subsection{Subsystem Codes}
\label{subs:subsystem}
Subsystem codes are stabilizer codes in which some logical qubits do not encode any information. The presence of these unused logical qubits a.k.a. {\em gauge qubits} often allows one to simplify encoding and decoding circuits since the initial and final state of the gauge qubits is irrelevant. A subsystem code derived from a stabilizer code $\calS$ is fully characterized by its {\em gauge group} $\calG$ that includes all stabilizers $\calS$ and all logical operators of $\calS$ acting only on the gauge qubits. The gauge group uniquely defines the stabilizer group (up to overall phase factors) such that $\la i\, I, \calS\ra =\calG\cap \calC(\calG)$. We shall often ignore the overall phase factors and use a simpler identity $\calS=\calG\cap \calC(\calG)$. A subsystem code has two species of logical operators. {\em Bare } logical operators are the logical operators of $\calS$ that act trivially on the gauge qubits. Multiplying bare logical operators by elements of $\calG$ one obtains {\em dressed} logical operators. Non-trivial bare and dressed logical operators are elements of $\calC(\calG)\backslash \calG$ and $\calC(\calS)\backslash \calG$ respectively. Errors are detected and corrected by measuring the gauge group generators $G_i$ and deducing the $\pm 1$ eigenvalues of the generators of ${\cal S}$ (the syndrome bits) from these measurements. These syndrome bits are then used to correct errors modulo operators in the gauge group ${\cal G}$ (which do not affect the encoded information), see Section~\ref{subs:honeycomb} and Appendix~A for more details.

\section{Topological Subsystem Codes}
\label{sec:tsc}

\subsection{Desired Code Properties}
\label{subs:properties}
The purpose of this section is to introduce a toolbox that allows one to design new 2D topological subsystem codes. The physical qubits of such codes occupy sites of a two-dimensional lattice, whereas the gauge group is generated by spatially local two-qubit generators. We begin by introducing some notation and describing desired code properties. Let $\calL=(V,E)$ be the underlying lattice embedded into a plane or a torus with a set of sites $V$ and a set of links (edges) $E$. Each site $u\in V$ represents a physical qubit. Pauli operators acting on $u$ will be denoted $X_u,Y_u$, and $Z_u$. Each link $(u,v)\in E$ represents a gauge group generator $K_{u,v}\in \calP_n$ acting on a pair of qubits $u,v$. We shall refer to the generators $K_{u,v}$ as {\em link operators}. For example, $X_u X_v$, $Y_u Y_v$, and $X_u Y_v$ are admissible link operators. For simplicity let us assume that the link operators are hermitian, so that $K_e^2=I$ for any link $e\in E$. Link operators generate a gauge group
\[
\calG=\la K_e\,  \: \, e\in E \ra,
\]
which in turn specifies a subsystem code as described in Section~\ref{subs:subsystem}. Let us  say that a Pauli operator $P\in \calP_n$ is {\em spatially local} iff the support of $P$ is confined to a disk of radius $O(1)$. Our goal is to construct subsystem codes combining the following properties:
\begin{enumerate}
\item[\bf (C1)] The stabilizer group $\calS=\calG\cap \calC(\calG)$ has spatially local generators. Elements of $\calS$ can be identified with closed homologically trivial loops on the lattice.
\item[\bf (C2)] Syndrome extraction can be implemented by measuring eigenvalues of the two-qubit  link operators.
\item[\bf (C3)] The code encodes one or more logical qubits. Bare logical operators (in ${\cal C}({\cal G})\backslash {\cal G}$) can be identified with closed homologically non-trivial loops.
\end{enumerate}
Properties (C1) and (C3) are often used to describe topological quantum order in the framework of stabilizer codes, see for instance~\cite{BHM:topo}. Property (C2) is the main advantage of subsystem codes compared to topological stabilizer codes such as the surface code. It says that the syndrome extraction requires only two-qubit measurements applied to nearest-neighbor sites of the lattice. In contrast, the syndrome extraction for all known topological stabilizer codes requires at least four-qubit measurements (it seems hard if not impossible to construct a nontrivial stabilizer code with 3- or 2-body, spatially local, commuting generators). We shall refer to subsystem codes satisfying conditions (C1-3) as {\em topological subsystem codes}.

\subsection{Kitaev's Honeycomb Model on $3$-valent Graphs}
\label{subs:honeycomb}
Our starting point is Kitaev's honeycomb model~\cite{Kitaev05} and its generalization to arbitrary $3$-valent graphs. It will provide us with a supply of `trivial' subsystem codes that satisfy conditions (C1) and (C2) but do not encode any logical qubits. These trivial codes will serve as a template for more sophisticated code constructions introduced later.

Recall that the honeycomb model~\cite{Kitaev05} is described by a Hamiltonian
\bea
H&=&-J_x \sum_{(u,v)\in E_x} X_u X_v -  J_y \sum_{(u,v)\in E_y} Y_u Y_v \nonumber \\
&& - J_z \sum_{(u,v)\in E_z} Z_u Z_v \label{kitaev}
\eea
where qubits occupy sites of the honeycomb lattice $\calL=(V,E)$  embedded into a torus with three species of edges, $E=E_x \cup E_y \cup E_z$, as illustrated in Figure~\ref{fig:honeycomb}. We shall regard the terms in the Hamiltonian Eq.~(\ref{kitaev}) as link operators of a subsystem code, that is,
\be
\label{kitaev1}
K_{u,v}=\left\{ \ba{rcl}
X_u X_v &\mbox{if} & (u,v)\in E_x, \\
Y_u Y_v &\mbox{if} & (u,v)\in E_y, \\
Z_u Z_v &\mbox{if} & (u,v)\in E_z. \\
\ea
\right.
\ee
\begin{figure}[t!]
\centering
\includegraphics[height=2cm]{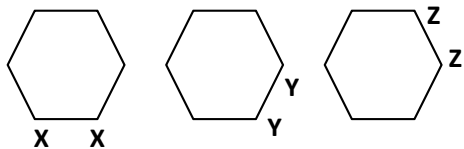}
\caption{Three types of link operators in the Kitaev's honeycomb model. This definition can be extended to the entire lattice using the translation invariance.}
\label{fig:honeycomb}
\end{figure}
Note that the link operators are not independent. Since each site has exactly three incident links contributing Pauli operators $X,Y$, and $Z$ acting on this site, one has an identity
\be
\label{kitaev2}
\prod_{e\in E} K_e \sim I.
\ee
Furthermore, a pair of  link operators $K_e$, $K_{e'}$ anticommutes iff the links $e$ and $e'$ share exactly one end-point and commute in all remaining cases, that is,
\be
\label{d3com}
K_e K_{e'} =(-1)^{\eta(e,e')}\, K_{e'} K_e \quad \mbox{for all $e,e'\in E$},
\ee
where
\[
\eta(e,e')=\left\{ \ba{rcl} 1 &\mbox{if} & \mbox{$e$ and $e'$ share exactly one site}, \\
0 && \mbox{otherwise}. \\
\ea
\right.
\]
Consider now  a {\em loop operator} $W(\gamma)=\prod_{e\in \gamma} K_e$ associated with any closed loop $\gamma \subseteq E$ on the lattice. Using Eq.~(\ref{d3com}) one can easily check that $W(\gamma)$ commutes with all link operators.  Thus, if we consider a subsystem code with a gauge group $\calG=\la K_e, \quad e\in E\ra$, then the loop operators are stabilizers of $\calG$, that is, $W(\gamma)\in \calS=\calG \cap \calC(\calG)$ for any closed loop $\gamma$. In fact, as we prove below (see Lemma~\ref{lemma:d3}) the stabilizer group $\calS$ is generated by the loop operators. It shows that the honeycomb model provides a natural supply of codes satisfying condition (C1).

Let us now discuss how the syndrome measurement can be implemented by measuring eigenvalues of individual link operators. Consider first the extraction of a single syndrome bit corresponding to some stabilizer $S\in \calS$. Suppose we can represent $S$ as an ordered product of link operators
\be
\label{syndrom0}
S=K_{m} \cdots K_{2} K_{1} \quad \mbox{where $K_j\equiv K_{e_j}$},
\ee
such that each $K_j$ commutes with the product of all link operators that appear before it, that is,
\be
\label{syndrom1}
[K_j,K_{j-1} \cdots K_1 ]=0 \quad \mbox{for all $j=2,\ldots,m$}.
\ee
Note that Eq.~(\ref{syndrom1}) implies that $S$ is a hermitian operator since the total number of anticommuting pairs among $K_1,\ldots,K_m$ is even. In order to measure the eigenvalue $\sigma$ of $S$, one can perform the nondemolition  eigenvalue measurement of $K_1,\ldots,K_m$ in the order from the left to the right, obtaining outcomes $\gamma_1,\ldots,\gamma_m\in \{+1,-1\}$. The deduced eigenvalue of $S$ is then computed as $\sigma=\gamma_1\cdots \gamma_m$.
%SBB4: small changes; simplified
Let us check that the above procedure is equivalent to the nondemolition  eigenvalue measurement of $S$ modulo gauge operators. For example, if $m=2$ then Eq.~(\ref{syndrom1}) implies that $K_1$ and $K_2$ mutually commute. Accordingly, the eigenvalues of $K_1$ and $K_2$ can be measured simultaneously. Since $S=K_2 K_1$, the eigenvalue of $S$ is determined by $\sigma=\gamma_1 \gamma_2$. We can now use induction to analyze the case $m>2$. Define a truncated version of the stabilizer $S$, namely, $S_j=K_j \cdots K_1$. Our induction hypothesis is that the reduced state obtained immediately after the measurement of $K_j$ is an eigenvector of $S_j$ with an eigenvalue $\sigma_j=\gamma_1 \cdots \gamma_j$. Indeed, we have already checked this for $j=1,2$.  Since Eq.~(\ref{syndrom1}) implies that $K_{j+1}$ commutes with $S_j$,  the eigenvalue measurement of $K_{j+1}$ preserves the eigenvalue of $S_j$. Hence the eigenvalue of $S_{j+1}$ after the measurement of $K_{j+1}$ is $\sigma_{j+1}=\gamma_{j+1} \sigma_j$. In Appendix~A we prove that the existence of the decomposition Eq.~(\ref{syndrom0}) satisfying Eq.~(\ref{syndrom1}) is in fact a necessary and sufficient condition under which the eigenvalue measurement of $S$ can be simulated by the measurements of the gauge group generators. This condition applies to any subsystem code and any fixed choice of the gauge group generators.

Consider now a loop $\gamma\subseteq E$ associated with the boundary of some hexagon. Let $e_1,e_2,\ldots,e_6$ be the six consecutive links comprising $\gamma$, and $K_j\equiv K_{e_j}$ be the corresponding link operators. In order to measure eigenvalue of the loop operator $W(\gamma)$, one can use a decomposition $W(\gamma)=K_6 K_4 K_2 K_5 K_3 K_1$, which clearly satisfies Eq.~(\ref{syndrom1}). In fact, since the $K_j$ pairwise commute for all odd (even) $j$, one can measure the eigenvalue of $W(\gamma)$ in only two rounds: $K_1,K_3,K_5$ are measured in the first round, and  $K_2,K_4,K_6$ are measured in the second round. Let us point out that the above procedure breaks down for loops $\gamma$ with an {\em odd length}. In fact, one can easily check that condition Eq.~(\ref{syndrom1}) cannot be satisfied for odd loops, and thus the corresponding loop operators $W(\gamma)$ cannot be measured using only two-qubit measurements. It shows that bipartiteness of the lattice $\calL$ plays a crucial role since for a bipartite graph any closed loop has  even length.

Syndrome measurements associated with a pair of adjacent hexagons may interfere with each other since one link operator is shared by both stabilizers. However one can arrange all the hexagons in $O(1)$ layers such that hexagons from the same layer do not have common links. Now the syndrome extractions for different hexagons in the same layer do not interfere with each other and one needs only two measurement rounds per each layer. The full syndrome extraction thus requires $O(1)$ rounds of two-qubit measurements~\footnote{For simplicity we ignored the non-local stabilizers associated with homologically non-trivial loops on the lattice. Using the same arguments as above, one can measure the eigenvalue of any such stabilizer using only two rounds of two-qubit measurements.}. It shows that the honeycomb model provides a natural supply of codes satisfying conditions (C1-2). This discussion does not address completely how efficient the measurement procedure is from a practical point of view, i.e. precisely how many rounds of measurements are needed.

Unfortunately, as we demonstrate below, the honeycomb model does not satisfy condition (C3), that is, the corresponding subsystem code has no logical qubits. To set the stage for our results, it will be convenient to analyze a generalization  of the honeycomb model Eqs.~(\ref{kitaev},\ref{kitaev1}) defined on an arbitrary $3$-valent graph $\calL=(V,E)$. As before, we assume that the physical qubits live at sites $u\in V$ and the link operators $K_e$  obey  the commutation rules Eq.~(\ref{d3com}). We note that the commutation rules Eq.~(\ref{d3com}) uniquely specify the link operators up to local unitaries. Indeed, let $u\in V$ be any site and $f,g,h\in E$ be the three links incident to $u$. From Eq.~(\ref{d3com}) we infer that the link operators $K_f,K_g$, and $K_h$ act on $u$ as $X_u$, $Y_u$, and $Z_u$ up to a permutation. Conjugating $K_f,K_g$, and $K_h$ by an appropriate single-qubit Clifford unitary operator $V_u$ one can generate any permutation of $X_u$, $Y_u$, and $Z_u$. In addition, the depolarizing noise model used in this paper  assigns equal probabilities to all Pauli errors on $u$, and thus it is also invariant under permutations of $X_u$, $Y_u$, and $Z_u$. It shows that the specific form of the link operators does not play any role as long as they obey the commutation rules Eq.~(\ref{d3com}). Let $\calG=\la K_e, \; e\in E\ra$ be the gauge group generated by the link operators.
\begin{dfn}
A subset of links $M \subseteq E$ is called a cycle iff every site has an even number of incident links from $M$.
\end{dfn}
\begin{lemma}
\label{lemma:d3}
The subsystem code $\calG$ has no logical qubits, that is, $\calC(\calG)\subseteq \calG$. Furthermore, $P\in \calC(\calG)$ iff $P\sim \prod_{e\in M} K_e$ for some cycle $M\subseteq E$.
\end{lemma}
\begin{proof}
It suffices to prove the second statement of the lemma. Consider any operator $P\in \calC(\calG)$. We shall say that $P$ {\em locally anticommutes} with a link operator $K_{u,v}$ iff
\[
P_u K_{u,v} = - K_{u,v} P_u \quad \mbox{and} \quad
P_v K_{u,v} = - K_{u,v} P_v.
\]
The subset $M\subseteq E$ including  all links locally anticommuting with $P\in \calC(\calG)$ will be called an {\em anticommuting set} of $P$. We claim that the anticommuting set $M$ is a cycle, that is, any site has even number of incident links from $M$. Indeed, Eq.~(\ref{d3com}) implies that for any site $u$ the product of $K_e$ over the three links $e$ incident to $u$ acts trivially on $u$. Hence $P$ can locally anticommute only with even number of $K_e$ among the three links $e$ incident to $u$. It shows that $M$ is a cycle. Consider now an operator
\[
P'=\prod_{e\in M} K_e.
\]
Using the fact that $M$ is a cycle and Eq.~(\ref{d3com}), one can easily check that $P'\in \calC(\calG)$. Furthermore, the anticommuting set of $P'$ is $M$. Hence $PP'$ has trivial anticommuting set, that is, $PP'$ locally commutes with any link operator $K_e$. It means that $PP'$ commutes with any single-qubit Pauli operator, which is possible only if $PP'\sim I$. Hence $P\sim P'$.
\end{proof}

\subsection{Loop Operators on $3$-valent Hypergraphs}
\label{subs:loops}
As we demonstrated in the previous section, the honeycomb model and its generalizations to $3$-valent graphs give rise to trivial subsystem codes encoding no logical qubits. In particular, any pair of loop operators commute with each other regardless of the homological class of the loops. Hence the codes derived from the honeycomb model cannot satisfy condition (C3). To overcome this problem, it is helpful to change the perspective and view as fundamental objects the loop operators rather than the link operators. Using the ideas from the previous section, we shall define a group of loop operators $\calG_{loop}$ associated with any $3$-valent hypergraph $\calL=(V,E)$ in which edges represent pairs and triples of sites. The latter type of edges  gives rise to anticommuting pairs of loop operators which could be used to encode logical qubits. We shall see that a pair of loop operators anticommute (commute) iff the two loops share an odd (even) number of triangles, see Lemma~\ref{lemma:loopCR}. Once the group of loop operators is defined, the gauge group $\calG$ of a subsystem code and its stabilizer group $\calS=\calG\cap \calC(\calG)$  can be reconstructed using identities
\be
\label{loops1}
\calG=\calC(\calG_{loop}) \quad \mbox{and} \quad \calS=\calG_{loop}\cap \calC(\calG_{loop}).
\ee
We will prove that $\calG$ has two-qubit generators for all hypergraphs $\calL$ obeying certain mild restrictions. The main advantage of this top-bottom approach is that once a suitable group of loop operators is chosen, we can easily find bare logical operators $\bar{X}_i,\bar{Z}_i$ and generators of the stabilizer group $\calS$. Indeed, Eq.~(\ref{loops1}) implies that
\be
\label{loops2}
\calC(\calG)=\calG_{loop},
\ee
that is, the bare logical operators $\bar{X}_i,\bar{Z}_i\in \calC(\calG)$ can be chosen as pairs of loop operators obeying the standard Pauli commutation rules, while the stabilizers are the loop operators commuting with all other loop operators. In contrast, deriving logical operators and stabilizers starting from the link operators may require rather lengthy calculations and does not give much insight into the topological nature of the code.

As before, we consider a system of $n$ physical qubits occupying a set of sites $V$ embedded into a plane or a torus. To define a subsystem code, we shall identify $V$ with a set of vertices of a hypergraph $\calL=(V,E)$ with two species of edges: edges connecting pairs of sites and edges connecting triples of sites. For consistency with the previous notations, we shall refer to the edges of the first type as {\em links} and the edges of the second type as {\em triangles}. Thus $E=E_2\cup E_3$, where $E_2$ is a set of links and $E_3$ is a set of triangles. More specifically, a link $(u,v)\in E_2$ is an unordered pair of sites and a triangle $(u,v,w)\in E_3$ is an unordered triple of sites. We shall say that an edge $e$ is incident to a site $u$ iff $e=(u,v)$ or $e=(u,v,w)$. There will be only three restrictions imposed on the hypergraph $\calL$, namely,
%SBB: perhaps some of these restrictions are not really necessary ?
\begin{itemize}
\item Each site has exactly three incident edges,
\item Any pair of edges share at most one site,
\item All triangles are pairwise disjoint.
\end{itemize}
We show examples of such hypergraphs $\calL$ in Figures~\ref{fig:so},\ref{fig:s5}. A different example is the triangle-honeycomb lattice studied in Ref.~\cite{yao_kivelson} but for this decorated honeycomb model $\calG_{loop}$ is trivial. By analogy with the honeycomb model, we consider a family of {\em edge operators} $K_e$ living on edges $e\in E$ and satisfying the commutation rules Eq.~(\ref{d3com}). We shall refer to operators $K_e$ as {\em link operators} if $e\in E_2$ and as {\em triangle operators} if $e\in E_3$. The same arguments as above show that the edge operators are uniquely defined by the commutation rules Eq.~(\ref{d3com}) up to local unitaries. Also the commutation rules Eq.~(\ref{d3com})  and our assumptions on $\calL$ imply that the triangle operators are three-qubit Pauli operators, for example, $K_{u,v,w}=Z_u Z_v Z_w$.
\begin{dfn}
A subset of edges $M\subseteq E$ is called a cycle iff every site has even number of incident edges from $M$.
\end{dfn}
Our assumptions on $\calL$ imply that any cycle $M$ consist of a disjoint union of loops, open paths, and triangles  such that free ends of every open path belong to some triangle, and every triangle $(u,v,w)\in M$ has exactly three open paths terminating at $u,v$, and $w$. See Figures~\ref{fig:so},\ref{fig:s5} for examples of a cycle. Given any cycle $M$, the operator
\be
\label{loop_operator}
W(M)=\prod_{e\in M} K_e
\ee
will be referred to as a {\em loop operator} associated with $M$. In contrast to the standard honeycomb model of Section~\ref{subs:honeycomb}, the loop operators $W(M)$ may anticommute with each other.
\begin{lemma}
\label{lemma:loopCR}
Let $e\in E$ be any edge and $M\subseteq E$ be any cycle. The loop operator $W(M)$ anticommutes with $K_e$ if $e$ is a triangle contained in $M$, i.e, $e\in M\cap E_3$, and commutes otherwise. Furthermore, for any cycles  $M,M'\subseteq E$ one has
\be
\label{loopCR}
W(M) W(M')=(-1)^{\Delta(M,M')}\, W(M') W(M),
\ee
where $\Delta(M,M')$ is the number of triangles shared by $M$ and $M'$.
\end{lemma}
\begin{proof}
Let $A_u(e)$ be the incidence matrix of $\calL$, that is, $A_u(e)=1$ if $e$ is an edge incident to $u$ and $A_u(e)=0$ otherwise. Recall that $K_e K_{e'} =(-1)^{\eta(e,e')}\, K_{e'} K_e$, where $\eta(e,e')=1$ iff $e$ and $e'$ share exactly one site and $\eta(e,e')=0$ otherwise. One can easily check that for any pair of edges $e,e'\in E$ one has the following identity:
\be
\label{eta1}
\eta(e,e')=\chi_3(e) \delta_{e,e'} + \sum_{u\in e} A_u(e') \pmod{2},
\ee
where $\chi_3(e)=1$ if $e$ is a triangle and $\chi_3(e)=0$ if $e$ is a link. For any site $u$ and for any cycle $M$ one has
\[
\sum_{e'\in M} A_u(e') = 0 \pmod{2}.
\]
Taking into account Eq.~(\ref{eta1}) one arrives at $K_e W(M) =(-1)^{\alpha(e)} W(M) K_e$, where
\bea
\alpha(e)&=&\sum_{e'\in M} \eta(e,e') \pmod{2} \nonumber \\
&=& \sum_{e'\in M} \chi_3(e) \delta_{e,e'} + \sum_{u\in e} A_u(e') \pmod{2}\nonumber \\
&=& \chi_3(e) \sum_{e'\in M}\delta_{e,e'} \pmod{2}. \nonumber
\eea
Hence $K_e$ anticommutes with $W(M)$ if $e\in M$ is a triangle and commutes otherwise. It immediately implies Eq.~(\ref{loopCR}).
\end{proof}
It is worth pointing out that the number of triangles in any cycle must be even, so Eq.~(\ref{loopCR}) is consistent with the fact that $W(M)$ commutes with itself. Indeed, each open path in $M$ contributes two sites that must be covered by triangles in $M$. Since each triangle covers exactly three sites, the number of triangles must be even.

\subsection{Reconstructing the Gauge Group from the Loop Operators}
\label{subs:loops2gauge}

Recall that our goal is to represent the group of loop operators $\calG_{loop}$ as a centralizer of the gauge group $\calG$ which has  two-qubit generators, see Eqs.~(\ref{loops1},\ref{loops2}). Let us first explain how to transform the hypergraph $\calL=(V,E)$ into an ordinary graph $\calL'=(V,E')$ which will be the underlying lattice for our subsystem code. The graph $\calL'$ has the same set of sites $V$ representing locations of the physical qubits. Recall that $E=E_2\cup E_3$, where $E_2$ is the set of links and $E_3$ is the set of triangles. Every link $(u,v)\in E_2$ becomes a link of the graph $\calL'$, that is, we have the inclusion $E_2\subseteq E'$. Furthermore, every triangle $(u,v,w)\in E_3$ contributes a triple of links to the graph $\calL'$, namely the links $(u,v)$, $(v,w)$, and $(u,w)$. In other words, the graph $\calL'$ is obtained from the hypergraph $\calL$ by breaking up all triangles into the individual links. To avoid confusions between various species of links and edges, we shall refer to the links of $\calL'$ originating from links of $\calL$ as {\em solid links}. The links of $\calL'$ originating from triangles of $\calL$ will be referred to as {\em dashed links}. Hence $E'=E_s\cup E_d$ where $E_s\equiv E_2$ and $E_d$ are the sets of solid and dashed links. The solid and dashed links will be rendered properly on all figures.

Now we are ready to define the gauge group $\calG$. As was mentioned above, the commutation rules Eq.~(\ref{d3com}) uniquely specify the edge operators  $K_e$ up to local unitaries. It will be convenient to fix the choice of these local unitaries such that all triangle operators include only Pauli $Z$, that is,
%SBB: dashed link operators will  be ZZ rather than YY for consistency with Ref [1].
\be
\label{triangle1}
K_{u,v,w} = Z_u Z_v Z_w \quad \mbox{for all $(u,v,w)\in E_3$}.
\ee
We define the set of link operators $K_e'$, $e\in E'$ generating the gauge group $\calG$ as follows:
\be
\label{K'}
\ba{rcll}
K_{u,v}' &=& K_{u,v} & \mbox{if $(u,v)$ is a solid link}, \\
K_{u,v}' &=& Z_u Z_v &\mbox{if $(u,v)$ is a dashed link}.\\
\ea
\ee
In other words, each triangle operator $K_{u,v,w}$ gives rise to a triple of dashed link operators obtained by restricting $K_{u,v,w}$ on pairs of qubits $(u,v)$, $(v,w)$, and $(u,w)$. (Note that these dashed link operators are not independent since $K_{u,v}' K_{v,w}' K_{u,w}'=I$). The main result of this section is the following lemma.
%SBB1: the missing part of the proof is added
\begin{lemma}
Let $\calG$ be the  group generated by solid and dashed link operators $K_e'$, $e\in E'$. Then $\calG_{loop}=\calC(\calG)$.
\end{lemma}
\begin{proof}
Let us first prove the inclusion  $\calG_{loop}\subseteq \calC(\calG)$. We need to check that for any link $e\in E'$ and for any cycle $M$ in the hypergraph $\calL$ the link operator $K_e'$ commutes with the loop operator $W(M)$.

Suppose $e\in E_s$ is a solid link. Then $K_e'=K_e$. Since $e$ is not a triangle, $K_e$ commutes with $W(M)$, see Lemma~\ref{lemma:loopCR}.

Suppose $e=(u,v)\in E_d$ is a dashed link so that $K_{u,v}=Z_u Z_v$. Let $f=(u,v,w)\in E_3$ be the corresponding triangle in $\calL$. Using the assumption Eq.~(\ref{triangle1}) and the commutation rules Eq.~(\ref{d3com}), we conclude that each of the sites $u,v,w$ has two incident links in $\calL$ and the corresponding link operators act on these sites by $X$ and $Y$. If $f\notin M$, each site  $u,v,w$ has either zero or two incident links from $M$. Hence $W(M)$ acts on $u,v,w$ by $I$ or $Z$. Since $K_{e}'=Z_u Z_v$, it follows that $K_e'$ commutes with $W(M)$. If $f\in M$, each site  $u,v,w$ has exactly one incident link from $M$. Hence $W(M)$ acts on $u,v,w$ by $X$ or $Y$. Again we conclude that $K_{e}'=Z_u Z_v$ commutes with $W(M)$.

It remains to prove the reverse inclusion $\calC(\calG)\subseteq \calG_{loop}$. The proof of this part is almost identical to the proof of Lemma~\ref{lemma:d3}. Consider any operator $P\in \calC(\calG)$. We claim that for any edge operator $K_e$, one has one of the two possibilities: (i) $P$ commutes with $K_e$ at every qubit in the support of $K_e$, or (ii) $P$ anticommutes with $K_e$ at every qubit in the support of $K_e$. Indeed, if $e=(u,v)\in E_2$ is a link then $K_e=K_e'\in \calG$ is a two-qubit Pauli operator commuting with $P$, that is, we have either (i) or (ii). If $e=(u,v,w)\in E_3$ is a triangle, the assumption $P\in \calC(\calG)$ means that $P$ commutes with $Z_u Z_v$, $Z_v Z_w$, and $Z_u Z_w$. Since $K_e=Z_u Z_v Z_w$, we have either (i) or (ii). We shall say that $P$ {\em locally commutes} with $K_e$ in the case~(i) and {\em locally anticommutes} with $K_e$ in the case~(ii). The subset $M\subseteq E$ including all edges locally anticommuting with $P\in \calC(\calG)$ will be called the {\em anticommuting set} of $P$.

We claim that the anticommuting set $M$ is a cycle, that is, any site $u$ has even number of incident edges from $M$. Indeed,  consider any site $u$ and let $f,g,h\in E$ be the triple of edges incident to $u$. Using  Eq.~(\ref{d3com}) and our assumptions on the hypergraph $\calL$, we  conclude that $K_f,K_g,K_h$ acts on $u$ as $X$, $Y$, $Z$ up to a permutation. It follows that $P_u$ anticommutes with even number of operators $K_f,K_g,K_h$. The observation made above shows that $P$ locally anticommutes with even number of $K_f,K_g,K_h$, that is, $M$ is a cycle.

Consider a loop operator $W(M)=\prod_{e\in M} K_e\in \calG_{loop}$. The inclusion $\calG_{loop}\subseteq \calC(\calG)$ implies $W(M)\in \calC(\calG)$. A direct inspection shows that the anticommuting set of $W(M)$ is $M$. Hence $P\cdot W(M)$ has trivial anticommuting set, that is, $P\cdot W(M)$ locally commutes with any edge operator $K_e$. It means that $P\cdot W(M)$ commutes with any single-qubit Pauli operator, that is, $P\cdot W(M)\sim I$ and $P\sim W(M)$. Hence $P\in \calG_{loop}$.
\end{proof}

\subsection{Examples of $3$-valent Hypergraphs}
\label{subs:examples}
In this section, we consider several explicit examples of $3$-valent hypergraphs and find the corresponding group of loop operators $\calG_{loop}$. As was argued above, the group $\calG_{loop}$ alone determines all properties of the code, so we shall not need the explicit form of the gauge generators.

Throughout this section we shall use a notation $\calZ(\calL)$ for the binary cycle space of a hypergraph $\calL=(V,E)$. In other words, elements of $\calZ(\calL)$  are cycles $M\subseteq E$ and the addition of cycles corresponds to the symmetric difference. Clearly, choosing generators of the group $\calG_{loop}$ is equivalent to choosing basis vectors of the binary cycle space $\calZ(\calL)$.

Let us first consider the hypergraph $\calL_{so}$ corresponding to the square-octagon lattice in Figure~\ref{fig:so}. Every site of $\calL_{so}$ has two incident links (solid lines) and one incident triangle (filled triangles). Let us first list basis cycles that have spatially local support (homologically trivial loops). Such cycles are associated with the elementary octagons and squares. Each octagon contributes two basis cycles $A,C \in \calZ(\calL_{so})$, see Figure~\ref{fig:so}. The cycle $A$ consists  of $8$ links forming the boundary of the octagon. The cycle $C$ is a loop of $8$ triangles and $12$ links encircling  the octagon (the links and triangles forming each cycle are highlighted in red in Figure~\ref{fig:so}). Analogously, each square contributes two basis cycles $B,D\in \calZ(\calL_{so})$. The cycle $D$ consists of  four links forming the boundary of the square. The cycle $B$ is a loop of $4$ triangles and $6$ links encircling  the square. Let us consider now the corresponding  loop operators, for example, $W(C)=\prod_{e\in C} K_e$. Note that each site $u$ has either $0$ or $2$ incident edges belonging to $C$. Commutation rules Eq.~(\ref{d3com}) imply that  if $u$ has two incident edges from $C$, the corresponding edge operators $K_e$ act on $u$ by different Pauli operators. It means that $W(C)$ acts non-trivially on a qubit $u$ iff $u$ has two incident edges from $C$. A direct inspection shows that $W(C)$ acts non-trivially on $24$ qubits. Similar arguments show that $W(A),W(B)$, and $W(D)$ act non-trivially on $8$, $12$, and $4$ qubits.

\begin{figure}[b!]
\centering
\includegraphics[width=.49\textwidth]{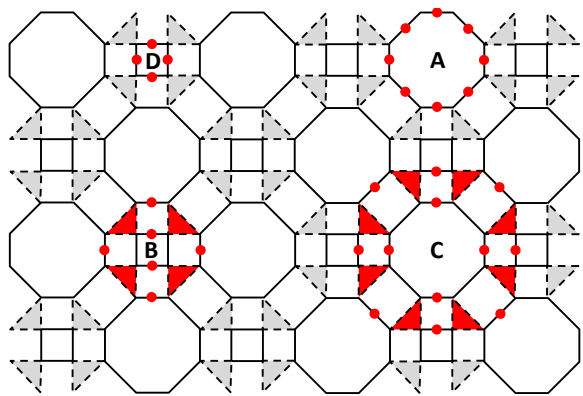}
\caption{(Color Online) The $3$-valent hypergraph $\calL_{so}=(V,E)$ associated with the square-octagon lattice embedded into a torus (periodic boundary conditions). The hypergraph has two types of edges: links (solid lines) and triangles (gray-filled and red-filled triangles). Each site has two incident links and one incident triangle. A cycle is a subset of edges $M\subseteq E$ such that each site has even number of incident edges from $M$. The hypergraph has four types of elementary cycles indicated by letters $A,B,C,D$. The links and triangles forming each cycle are highlighted in red. The underlying lattice of the subsystem code is obtained from the hypergraph $\calL_{so}$ by breaking up each triangle into a triple of links (dashed lines).}
\label{fig:so}
\end{figure}

The remaining basis cycles are those associated with the homologically non-trivial loops on the lattice. We denote these cycles $Z_1,Z_2,X_1$, and $X_2$ see Figure~\ref{fig:soZ} and Figure~\ref{fig:soX}. A simple algebra shows that the cycles $A,B,C,D$ and  $Z_1,Z_2,X_1,X_2$ form a basis of the cycle space $\calZ(\calL_{so})$.

\begin{figure}[t!]
\centering
\includegraphics[width=.49\textwidth]{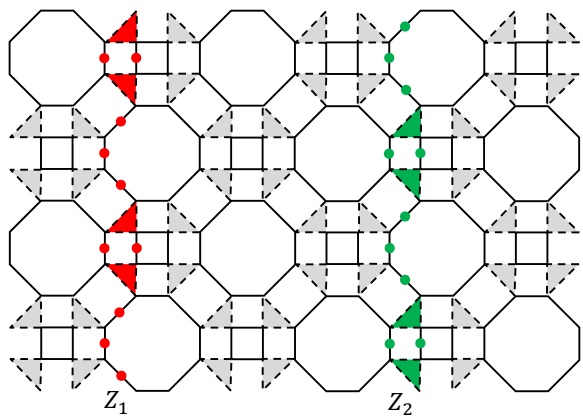}
\caption{(Color Online) Basis cycles $Z_1,Z_2\in \calZ(\calL_{so})$ corresponding to homologically non-trivial vertical loops on the lattice.}
\label{fig:soZ}
\end{figure}

\begin{figure}[h!]
\centering
\includegraphics[width=.49\textwidth]{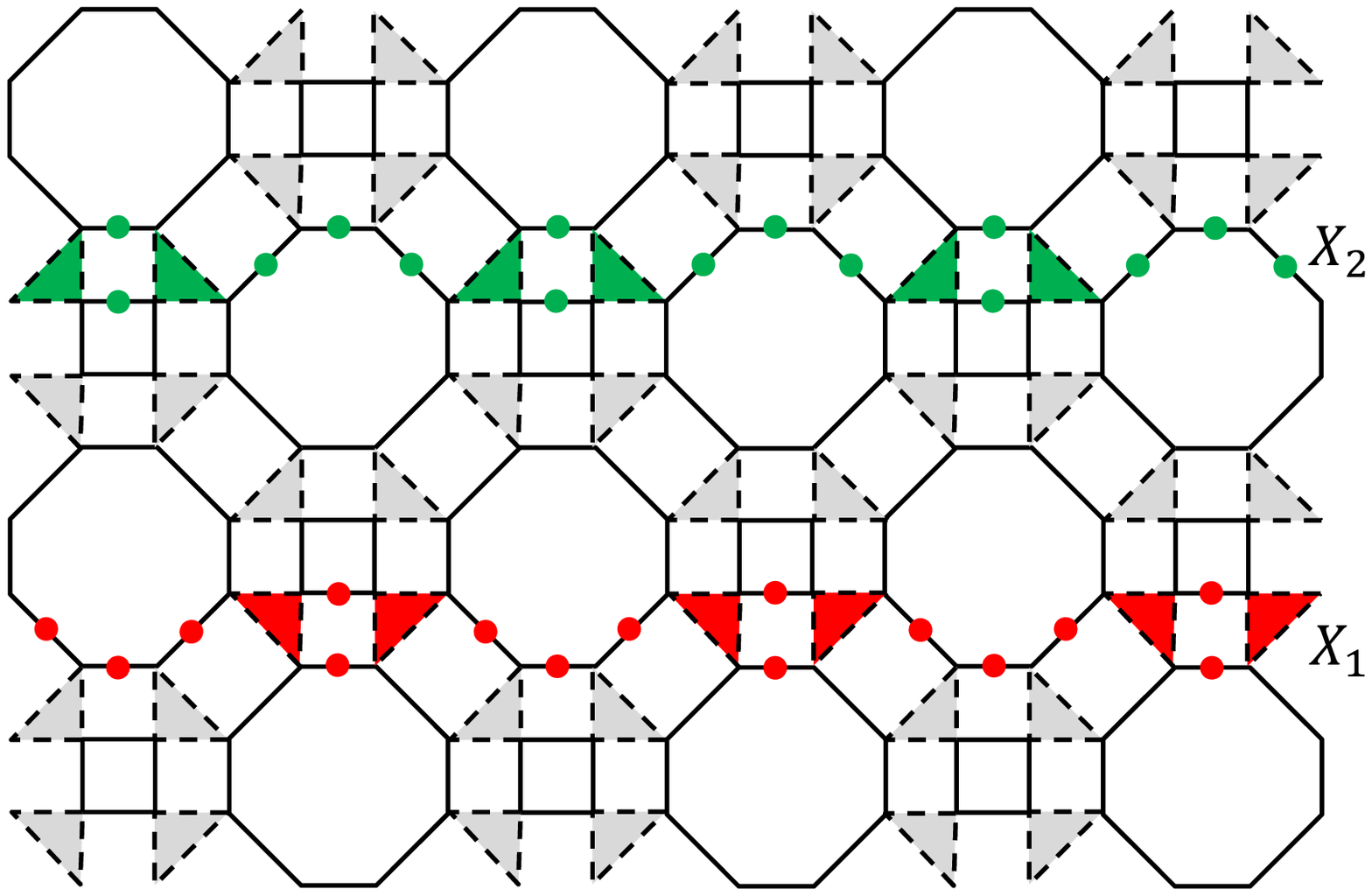}
\caption{(Color Online) Basis cycles $X_1,X_2\in \calZ(\calL_{so})$ corresponding to homologically non-trivial horizontal  loops on the lattice.}
\label{fig:soX}
\end{figure}

Let us now consider commutation rules between the corresponding loop operators. We note that the cycles $A$ and $D$ contain no triangles which implies that $W(A)$ and $W(D)$ commute with all loop operators, see Lemma~\ref{lemma:loopCR}, that is, $W(A),W(D)\in \calG_{loop}\cap \calC(\calG_{loop})$. Since $\calG_{loop}\cap \calC(\calG_{loop})=\calS$, see Eq.~(\ref{loops1}), we conclude that $W(A)$ and $W(D)$ are stabilizers. Furthermore, one can easily check that any cycle of type $B$ or $C$ shares either zero or two triangles with any other cycle $B,C,Z_1,Z_2,X_1,X_2$. Lemma~\ref{lemma:loopCR} then implies that $W(B)$ and $W(C)$ commute with all loop operators, that is, $W(B),W(C)\in \calS$. Thus all the loop operators associated with $A,B,C$, and $D$ cycles are stabilizers.

The cycles $Z_1,Z_2$ share exactly one triangle with the cycles $X_1,X_2$ respectively, see Figures~\ref{fig:soZ},\ref{fig:soX}. Lemma~\ref{lemma:loopCR} implies that
\be
\label{bare}
W(X_a) W(Z_b)=(-1)^{\delta_{a,b}} W(Z_b) W(X_a), \quad a,b=1,2.
\ee
We conclude that the code $\calG=\calC(\calG_{loop})$ has exactly two logical qubits and the loop operators $W(Z_1),W(Z_2)$, $W(X_1),W(X_2)$ can be chosen as the bare logical operators $\bar{Z}_1,\bar{Z}_2$, $\bar{X}_1$, $\bar{X}_2$ respectively.

Let us mention that the subsystem code $\calG=\calC(\calG_{loop})$ coincides with the topological subsystem code of~\cite{bombin:topsub} associated with the Union Jack lattice (the dual of the square-octagon lattice). Hence applying our construction to the square-octagon lattice does not provide new examples of codes but rather provides a new angle to look at the already known code.

Let us now consider the hypergraph $\calL_{5s}$ corresponding to the lattice shown in Figure~\ref{fig:s5}. We shall refer to this lattice as the five-squares lattice since it unit cell consists of  five squares. In contrast to the square-octagon lattice, some sites of $\calL_{5s}$ have three incident links and others have two incident links and one incident triangle. Simple algebra shows that the cycle space $\calZ(\calL_{5s})$ has four types of local basis cycles $A,B,C,D$ associated with loops encircling squares and octagons, see Figure~\ref{fig:s5}, and  homologically non-trivial basis cycles $Z_1,Z_2,X_1,X_2$, see Figure~\ref{fig:s5Z} and Figure~\ref{fig:s5X}. As before, Lemma~\ref{lemma:loopCR} implies that the loop operators $W(A),W(B),W(C),W(D)$ generate the stabilizer group $\calS=\calG_{loop}\cap \calC(\calG_{loop})$. The loop operators associated with cycles $Z_1,Z_2,X_1,X_2$ obey commutation rules Eq.~(\ref{bare}) and thus can be chosen as bare logical operators on the two logical qubits.

\begin{figure}[b!]
\centering
\includegraphics[width=.49\textwidth]{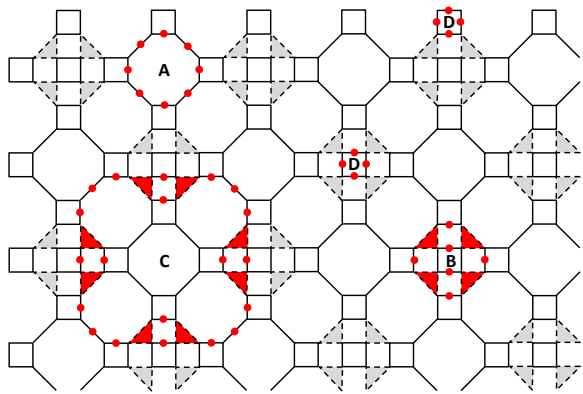}
\caption{(Color Online) The $3$-valent hypergraph $\calL_{5s}=(V,E)$ associated with the five-squares lattice embedded into a torus (periodic boundary conditions). The underlying lattice of the subsystem code is obtained from the hypergraph $\calL_{5s}$ by breaking up each triangle into a triple of links (dashed lines). Four type of local basis cycles $A,B,C,D\in \calZ(\calL_{so})$ are highlighted in red.}
\label{fig:s5}
\end{figure}

\begin{figure}[t!]
\centering
\includegraphics[width=.49\textwidth]{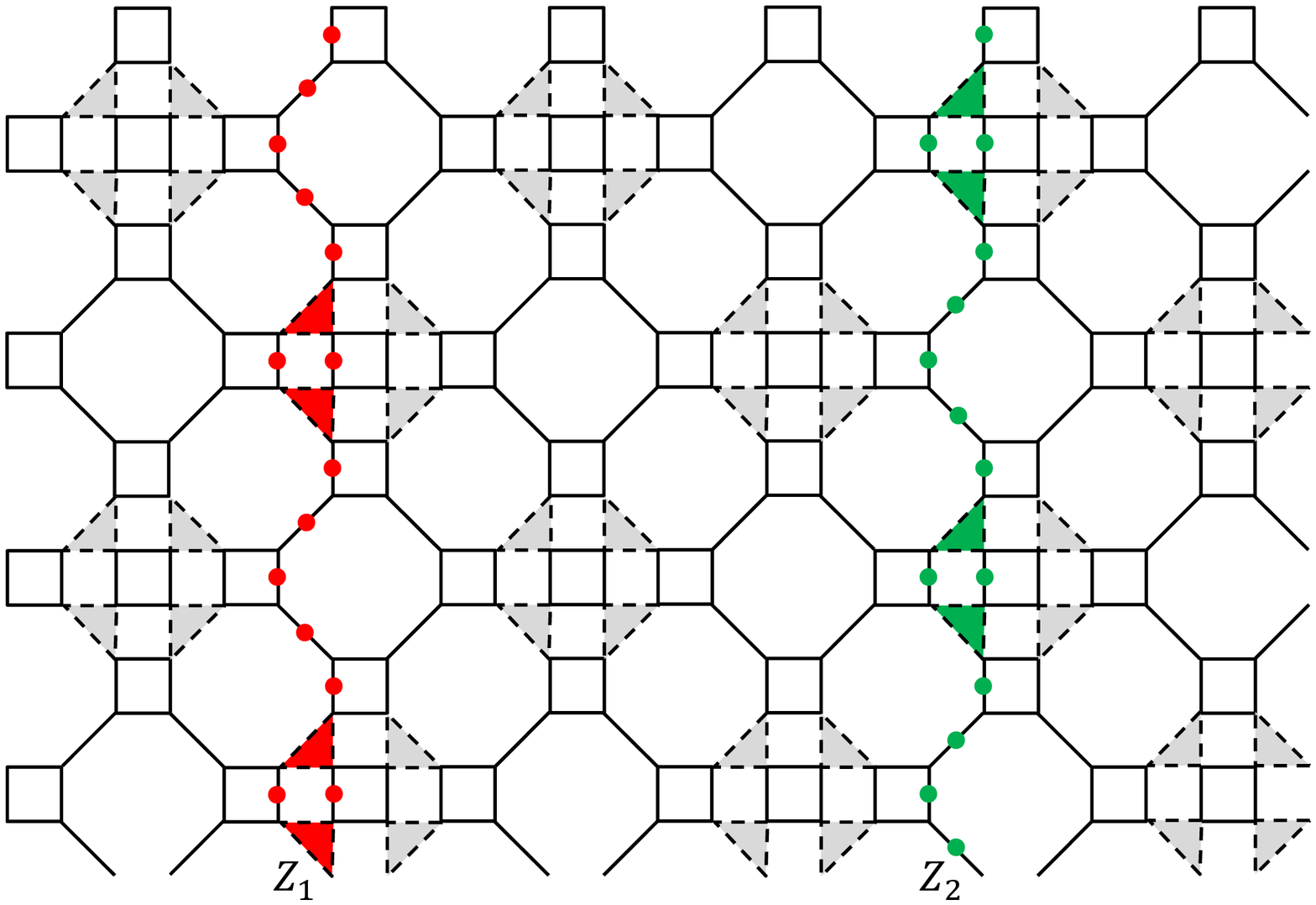}
\caption{(Color Online) Basis cycles $Z_1,Z_2\in \calZ(\calL_{5s})$ corresponding to homologically non-trivial vertical loops on the lattice.}
\label{fig:s5Z}
\end{figure}

\begin{figure}[h!]
\centering
\includegraphics[width=.49\textwidth]{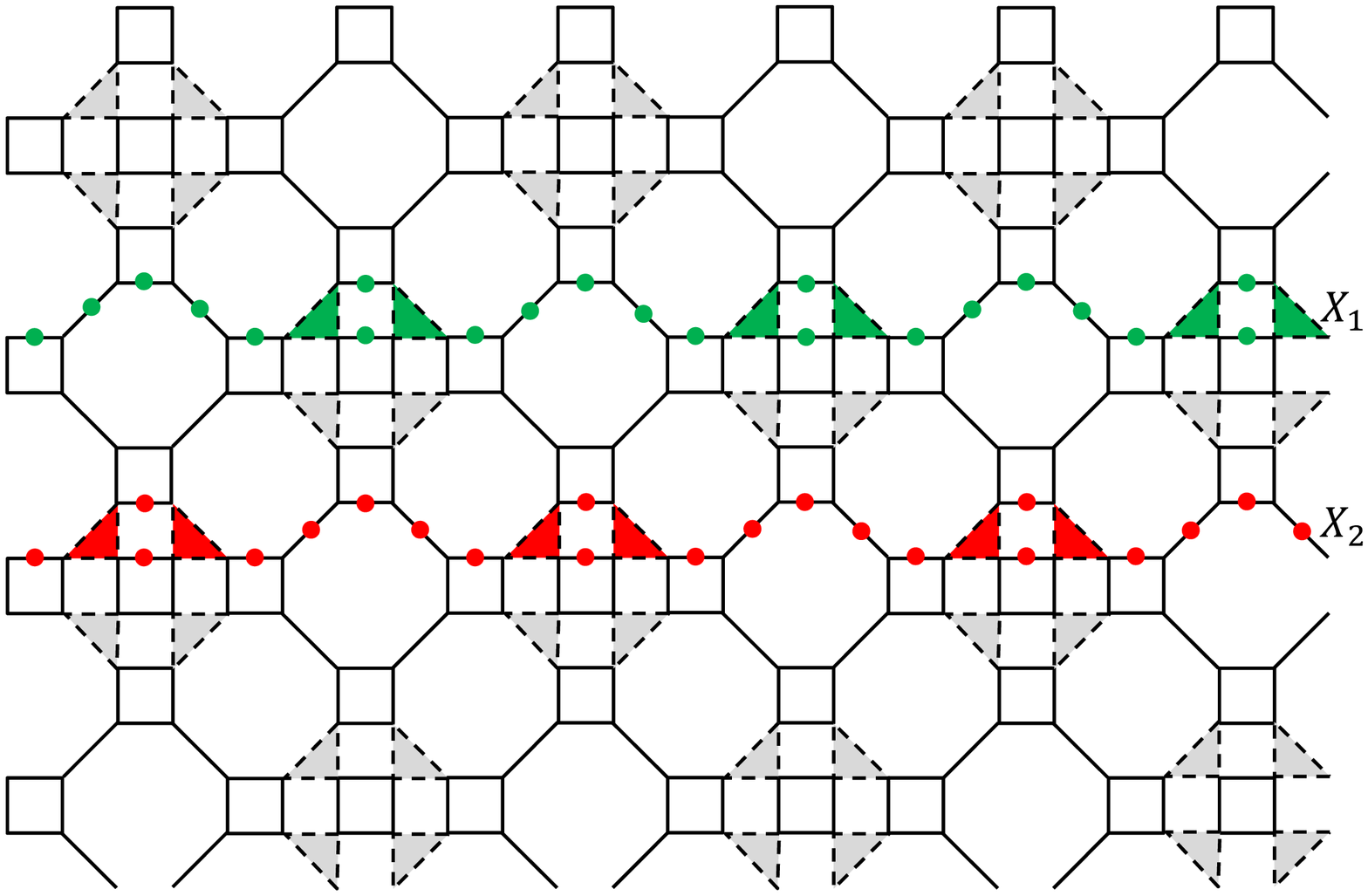}
\caption{(Color Online) Basis cycles $X_1,X_2\in \calZ(\calL_{5s})$ corresponding to homologically non-trivial horizontal loops on the lattice.}
\label{fig:s5X}
\end{figure}

%How do we measure the syndrome
Let us briefly sketch how the syndrome measurement for the stabilizers of different type can be simulated by measuring eigenvalue of the two-qubit link operators $K_e'$. Stabilizers of type $A$ and $D$ correspond to closed even-length loops containing no triangles. Hence the syndrome measurement strategy used for the honeycomb model in Section~\ref{subs:honeycomb} can be applied directly to stabilizers of type $A$ and $D$. Consider now a stabilizer of type $B$. It consists of a solid center square and a length-$8$ external loop with alternating solid and dashed links, see Figures~\ref{fig:so},\ref{fig:s5}. One can easily check that the stabilizer $W(B)$ can be represented as a product of the eight link operators $K_e'$ lying on the external loop and operators $K_e'$ associated with the two vertical solid links $e\notin B$ lying on the central square, see Figure~\ref{fig:stabilizersMeasurement}. In order to satisfy condition Eq.~(\ref{syndrom1}) one can divide  these $10$ links into three levels: Level-$1$: horizontal solid links, Level-$2$: vertical solid links, and  Level-$3$:  dashed links. One can easily check that all link operators $K_e'$ at the same level mutually commute and any link operator at level $j$ commutes with the product of all link operators at levels $<j$. Hence we can satisfy Eq.~(\ref{syndrom1}) using any ``level-nondecreasing" order of measurements: we first measure all level-$1$ links (in any order), then measure all level-$2$ links (in any order), and finally measure all level-$3$ links (in any order). A similar strategy can be applied to stabilizers of type $C$. For example, the stabilizer $W(C)$ for the five-squares code, see Figure~\ref{fig:s5}, can be represented as  a product of $40$ links operators $K_e'$. These $40$ links can be divided into three levels such that all link operators $K_e'$ at the same level mutually commute and any link operator at level $j$ commutes with the product of all link operators at levels $<j$. Hence any level-nondecreasing order of measurements obeys the condition Eq.~(\ref{syndrom1}).

\section{Decoding Algorithms}
\label{Decoding_algorithm}

% BMT: I inserted Sergey's note on max likelihood decoding and added some more
% things
% SBB2: do we really need the discussion on maximum likelihood decoding? It seems that we do not have any numerical results
% on that...
\subsection{Maximum Likelihood Decoding}
Here we describe the maximum likelihood decoding  for an arbitrary subsystem code $\calG$. Our derivation mostly follows \cite{bombin:topsub}. Let us assume that the noise can be described by random Pauli errors $E$ drawn from some distribution $\prob{E}$. Let $\calS=\calG\cap \calC(\calG)$ be the stabilizer group of $\calG$. For any Pauli error $E$ let $[E]$ be the (left) coset of $\calG$ that contains $E$, that is, $[E]=E\cdot \calG$. Obviously, the syndrome $s$ caused by an error $E$ depends only on the coset $[E]$. Let $s(h)$ be the syndrome caused by any error in the coset $h$. The goal of the decoding algorithm is to identify an error up to a gauge operator based on its syndrome.  In other words, one needs to guess the coset $h$ based on the syndrome $s(h)$.
\begin{prop}
Suppose the code has $k$ logical qubits.  Then for any given syndrome $s$ there are $4^k$ different cosets $h$ such that $s(h)=s$.
\end{prop}
\begin{proof}
Indeed, suppose $h=[E]$ and $h'=[E']$ such that $s(h)=s(h')=s$. Then $E\cdot E'\in \calC(\calS)$. Using the identity $\calC(\calS)=\calC(\calG)\cdot \calG$ we conclude that $E\cdot E'=\bar{P} G$ for some bare logical operator $\bar{P}\in \calC(\calG)$ (may be the identity), and some gauge operator $G\in \calG$. Since $E$ and $E'$ are defined modulo gauge operators, we can ignore $G$. Furthermore, if $h$ and $h'$ are distinct cosets then $\bar{P}$ is a non-trivial logical operator. Since there are exactly $4^k$ distinct bare logical operators, there must be exactly $4^k$ distinct cosets $h$ such that $s(h)=s$.
\end{proof}
Let $\calE(s)$ be the set of all $4^k$ cosets $h$ such that $s(h)=s$. The chosen probability distribution of errors induces some probability distribution $P(h)$  of cosets, namely,
\[
P(h)=\sum_{E\in h}\, \prob{E}.
\]
Given a syndrome $s$,  the maximum likelihood decoder computes the probability $P(h)$ for all cosets $h\in \calE(s)$ and then selects the coset with the largest probability. For any error $E$ one has
\[
\prob{[E]}=\sum_{G\in \calG} \prob{EG}.
\]
For example, for depolarizing noise with an error rate $p$, one has
\[
\prob{E}\sim \left(\frac{p}{3(1-p)}\right)^{|E|}.
\]
Let us introduce an inverse temperature $\beta$ such that
\be
\exp{(-2\beta)}=\frac{p}{3(1-p)}.
\label{eq:pT}
\ee
Then one arrives at
\[
\prob{[E]}\sim \sum_{G\in \calG} e^{-2\beta |E G|}\equiv Z(E).
\]
Note that by definition the partition function $Z(E)$ depends only on the coset $[E]$. Here $Z(E)$ can be regarded as a partition function of a classical spin Hamiltonian with a quenched disorder determined by $E$ and the inverse temperature $2\beta$. The classical spins are associated with generators $K_e$ of the gauge group. Indeed, we can represent
\[
G=\prod_{e\in E} K_e^{\alpha(e)}, \quad \alpha(e)\in \{0,1\}.
\]
Then the Hamming weight $|E G|$ can be expressed a local Hamiltonian depending on the classical binary variables $\alpha(e)$. Parameterizing $\alpha$'s by Ising spins, one express $Z(E)$ as a partition function of an Ising-like Hamiltonian with inverse temperature $\beta$. Finding the values of the binary variables $\alpha(e)$ which minimize the weight $|EG|$ corresponds to minimizing the energy of the local Hamiltonian. Due to the particular structure of the Pauli group, this minimization is an instance of a MAX-XORSAT problem with local SAT constraints between binary variables on a 2D lattice.

% notation L was undefined; say that this is not rigorous
Following~\cite{topological_quantum_memory,bombin:topsub} we can relate the threshold value of the error rate $p$ to a phase transition in the statistical mechanics model corresponding to $Z(E)$, where $E$ is regarded as a quenched disorder. Indeed, suppose the error rate $p$ is below the threshold and let $E$ be a typical error with a syndrome $s=s(E)$. Then one should expect that in the limit of large lattice size $L$ all cosets in $\calE(s)$ have negligible probability except for the coset $[E]$. In other words, $\prob{[E]}\gg \prob{[\bar{P} E]}$ for any non-trivial logical operator $\bar{P}\in \calC(\calG)\backslash \calG$.
% say about the case p<p_c
On the other hand, if the error rate $p$ is above the threshold, one should expect that in the limit of large lattice size all cosets in $\calE(s)$ have the same probability. Thus we can determine the threshold error rate $p_c$ using the following heuristic criterion:
\bea
p<p_c \quad \Rightarrow \quad
\lim_{L\to \infty} \sum_E \prob{E} \, \log{\left(\frac{Z(E)}{Z(\bar{P} E)}\right)} = \infty, \nonumber \\
p>p_c \quad \Rightarrow \quad
\lim_{L\to \infty} \sum_E \prob{E} \, \log{\left(\frac{Z(E)}{Z(\bar{P} E)}\right)} = 0.
\label{eq:belowabove}
\eea
% we still need to explain what phase transition did we mean above
Assuming that the code $\calG$ has a macroscopic distance, the logical operator $\bar{P}$ can be ``cleaned out" from any local region of the lattice by a gauge transformation~\cite{BT08}. Hence the difference between the partition functions $Z(E)$ and $Z(\bar{P} E)$ is hidden if one inspects any local region of the lattice. It can only reveal itself if one inspects the entire lattice and if the system exhibits a long-range order. Hence one should expect that the point $p=p_c$
corresponds to a phase transition between a magnetically ordered phase ($p<p_c$)
and a paramagnetic phase ($p>p_c$).
%SBB: if we submit to IEEE, the above discussion can be omitted (as well as the entire section on ML decoder)

In general, implementing maximum likelihood decoding is a computationally inefficient task, since one would have to estimate the different partition functions $Z(\bar{P}E)$ for a local 2D Hamiltonian with quenched disorder. Thus one has to seek efficient strategies which approximate the action of such ideal decoder. Given the relatively low error rates $p$ which correspond to low temperatures via Eq.~(\ref{eq:pT}), one may argue (as in \cite{topological_quantum_memory}) that maximum-likelihood decoding is pretty well approximated by considering $\beta \rightarrow \infty$ or doing {\em minimum-weight decoding}. For minimum-weight decoding, we consider which 2D local Hamiltonian with quenched disorder determined by $\bar{P}E$ has minimum energy. In other words, we find the error consistent with the syndrome which has minimum Hamming weight. Again, for this minimization problem, --instances of a two-dimensional MAX-XORSAT problem--, there does not necessarily exist an exact efficient algorithm. For the 2D surface code, there is an efficient algorithm for minimum-weight decoding which is Edmond's efficient polynomial time (in the lattice size) minimum-weight matching algorithm \cite{chinese_postman_problem}. Of course, even the absence of an exact algorithm does not preclude the existence of good decoders based on heuristic optimization strategies which may be more computationally efficient than exact algorithms. This is the motivation for the work in \cite{DCP:topdec_prl,DCP:topdec_long,PBGC:square_oct} in which the authors develop fast renormalization group decoding algorithms which are suitable for any topological stabilizer or subsystem code.

In Sections~\ref{ST_decoding} and \ref{Improved_ST_decoding} we will describe a simple decoding algorithm for the five-squares code. The decoding algorithm in Section \ref{ST_decoding} is improved in Section \ref{Improved_ST_decoding} by reducing the weight of the guessed error which gives a slight boost to the noise threshold. For the five-squares code, we have developed a decoder which is fine-tuned to the particular structure of the code and which maps the decoding after some preprocessing onto decoding of two copies of the toric code. For the decoding of the toric code, one can then use a minimum-weight matching algorithm or renormalization group decoding. Let us point out that applying a similar decoding strategy to the subsystem code defined on the square-octagon lattice, see Figure~\ref{fig:so}, one can reduce the decoding to correcting $Z$ errors using stabilizers of the topological color code~\cite{BMD:topo}. In \cite{PBGC:square_oct} the authors are able to establish a much more general correspondence between the toric code and any topological subsystem code \cite{bombin:topsub}. They show that any topological subsystem code is locally unitarily equivalent to some copies of the toric code. This implies that the decoding of any topological subsystem code can proceed via efficient decoding algorithms for the toric code.

\subsection{Decoding for the Five-Squares Code}
\label{ST_decoding}
% General high-level approach
To describe our decoding algorithm it will be convenient to fix a particular choice of the link operators $K_e$ associated with the solid links. We will choose $K_{u,v}=X_u Y_v$ for all solid links $(u,v)$ that belong to the boundary of some square, see Figure~\ref{fig:xy}. Commutation rules Eq.~(\ref{d3com}) then imply that one has to choose $K_{u,v}=Z_u Z_v$ whenever $(u,v)$ is a solid link such that $u$ and $v$ belong to different squares, see Figure~\ref{fig:zz}. Recall that the link operators associated with dashed links are fixed by Eq.~(\ref{K'}).

\begin{figure}[b!]
\centering
\includegraphics[width=.49\textwidth]{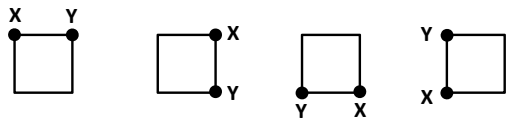}
\caption{Link operators $K_{u,v}$ associated with solid links $(u,v)$ that belong to boundary of some square.}
\label{fig:xy}
\end{figure}

\begin{figure}[b!]
\centering
\includegraphics[width=.49\textwidth]{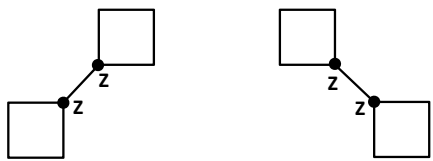}
\caption{Link operators $K_{u,v}$ associated with solid links $(u,v)$ connecting two different squares.}
\label{fig:zz}
\end{figure}

Our decoding algorithm exploits the fact that errors only need to be corrected up to multiplication with elements in the gauge group ${\cal G}$.  In particular, gauge group elements can be used to move errors around and simplify the decoding.
To illustrate this point let us consider any solid  square with the corresponding qubits labeled by $1,2,3,4$, see Figure~\ref{fig:xErrors}. We claim that instead of considering general Pauli errors acting on these qubits, it suffices to consider simplified errors from the group $\la X_1,Z_1,Z_2,Z_3,Z_4\ra$, that is, arbitrary $Z$ errors and $X$ errors acting only on qubit $1$. Indeed, using the gauge operators $K_{u,v}=X_u Y_v$ and $Z$ errors one can generate any Pauli operator $P$ that contains $Z$'s and even number of $X$'s. Multiplying $P$ by $X_1$ we can obtain any Pauli operator on qubits $1,2,3,4$. Thus it suffices to consider arbitrary $Z$ errors and $X$ errors acting only at a single location in each square (say, the north-west corner). Let us also note that the $D$-stabilizer associated with the considered square is
\[
S=K_{1,2} K_{2,3} K_{3,4} K_{4,1}= Z_1 Z_2 Z_3 Z_4.
\]

% Figure - location of the errors
\begin{figure}[t]
\centering
\subfigure[It is sufficient to correct $X$ errors at qubit $1$.]{
\includegraphics[width=.22\textwidth]{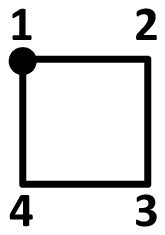}
\label{fig:xErrors}
}
\hspace{0.6mm}
\centering
\subfigure[$Z$ errors need to be corrected at qubits $1$, $2$, $4$, $19$, and $20$.]{
\includegraphics[width=.22\textwidth]{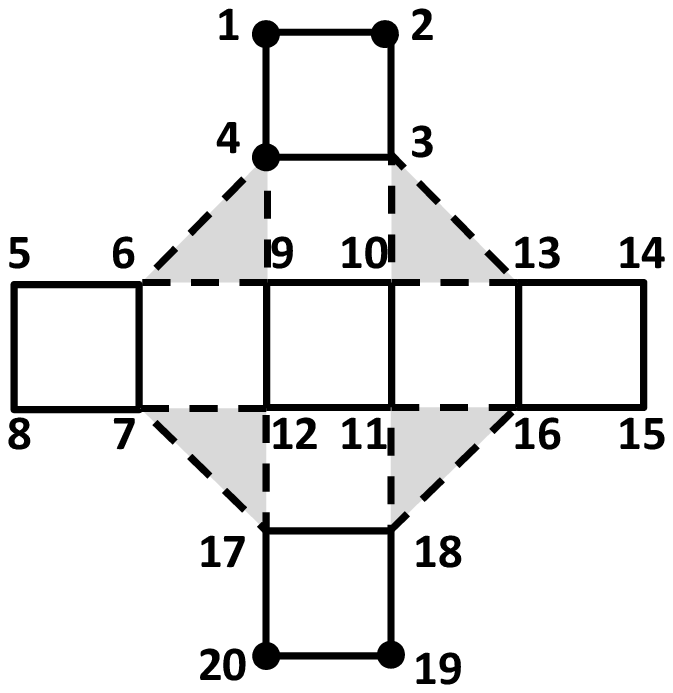}
\label{fig:yErrors}
}
\caption{Locations of the errors.}
\label{fig:errorLocations}
\end{figure}

We can now use a similar trick to reduce the set of possible locations for $Z$ errors. Consider a {\em unit cell} of the lattice shown in Figure~\ref{fig:yErrors} that consists of $20$ qubits. We claim that any $Z$ error acting on the unit cell can be reduced to $Z$ errors at the locations $1,2,4,19,20$. Indeed, let us first describe relevant gauge operators of $Z$ type that we shall use to simplify the error. Define a group $\calG^z\subseteq \la Z_1,\ldots,Z_{20}\ra$ generated by the link operators $Z_u Z_v$ associated with $12$ dashed links within the unit cell and $D$-stabilizers $Z_s Z_t Z_u Z_v$ associated with the four  squares on the top, bottom, left, and right edges of the cell. Let us partition the unit cell into $12$ {\em internal} qubits that belong to some triangle and $8$ {\em external} qubits, that is, qubits $1,2,5,8,14,15,19,20$, see Figure~\ref{fig:yErrors}. Consider a subgroup $\calG^z_{int}$ generated by the dashed link operators $Z_u Z_v$ and truncated versions of the $D$-stabilizers  $Z_3 Z_4$, $Z_6 Z_7$, $Z_{13} Z_{16}$, and $Z_{17} Z_{18}$. In other words, $\calG^z_{int}$ is a `projection' of the group $\calG^z$ onto the internal qubits. Since any pair of internal qubits can be connected by a path of $ZZ$ generators belonging to $\calG^z_{int}$, we conclude that  $\calG^z_{int}$ contains all {\em even-weight} $Z$-operators on the internal qubits. It follows that any $Z$ error $P$ acting on the unit cell can be represented as $P=P'G$, where $G\in \calG^z$ and $P'$ contains only $Z$ errors at the external qubits and some selected internal qubit, say qubit $4$. Finally, we can move $Z$ errors out of any external qubit $u\in \{5,8,14,15\}$ using the solid link operators $Z_u Z_v$, where $v\in \{1,2,19,20\}$ belongs to one of the adjacent unit cells.

To summarize, it suffices to consider only $X$ errors acting on the north-west corners of each square and $Z$ errors acting on locations $1,2,4,19,20$ of each unit cell, see Figures~\ref{fig:xErrors},\ref{fig:yErrors}.

We choose the following decoding procedure to correct this reduced set of errors. Assume that all stabilizer generators $A$, $B$ and $C$ and $D$ have been measured perfectly. We first consider the syndromes of stabilizers  $D$. Note that $D$ anti-commutes with any single $X$ error at the corresponding square and commutes with all $Z$ errors. For all non-zero $D$-syndromes, we will apply $X$ to the north-west corner of the corresponding square which corrects all $X$ errors.
After this step we are left with $Z$ errors at the locations $\{1,2,4,19,20\}$ in each unit cell.

Note that this error correction step may change the syndromes of $A$, $B$, and $C$. For example, an $X$ error in the square in the center of the unit cell anticommutes with the corresponding stabilizer $B$. The syndromes $A$, $B$, and $C$ can be either measured after the first error correction step, or their new values can be calculated based on their initial values and the error corrections that were applied (i.e. updating the so-called Pauli frame).

Any $B$ stabilizer acts trivially on locations $\{1,2,19,20\}$ in the corresponding unit cell and acts by $Y$ at the location $4$. Hence the only $Z$ error in the unit cell that can create a non-trivial $B$-syndrome is the error $Z_4$. For all non-zero $B$-syndromes we will apply $Z_4$ in the corresponding unit cell to correct the $Z$-error.

% Transform the lattice
It remains to show how to correct the remaining $Z$ errors at locations $\{1,2,19,20\}$. These qubits are divided into two groups, which are depicted separately in Figure~\ref{fig:transformedLattice1} and \ref{fig:transformedLattice2}. These figures show how the stabilizers $A$ and $C$ act on the qubits. Each of the stabilizers acts on four of the considered qubits. Stabilizers of type $A$ are represented by blue and orange squares, and $C$ operators are represented by red and green squares. Thus the reduced lattice effectively consists of two disjoint lattices on which the stabilizers acts as four-body plaquette operators. The only occurring errors on these two disjoint lattice are $Z$ errors which anti-commute individually with the $A$ and $C$ stabilizers which cover their support. Hence, one can view the lattice in e.g. Figure~\ref{fig:transformedLattice1} as half of a toric code and the $Z$ errors can be corrected using any decoding strategy for the toric code. We have used decoding via the minimum weight perfect matching algorithm similarly as in~\cite{topological_quantum_memory}.

%In contrast to the toric code, the minimum weight decoding for the color code requires computing the minimum energy of the three-spin Ising model~\cite{error_threshold_color}. This problem is known to be NP-hard~\cite{TK10} in the worst case. It gives strong evidence that the minimum weight decoding is computationally hard for the color codes. Hence one advantage  of the main advantage of the five-squares code is that it can be decoded using already available efficient algorithms and can be regarded a subsystem version of the toric code.

% Figure - the transformed lattice
\begin{figure}[t!]
\centering
\includegraphics[width=.49\textwidth]{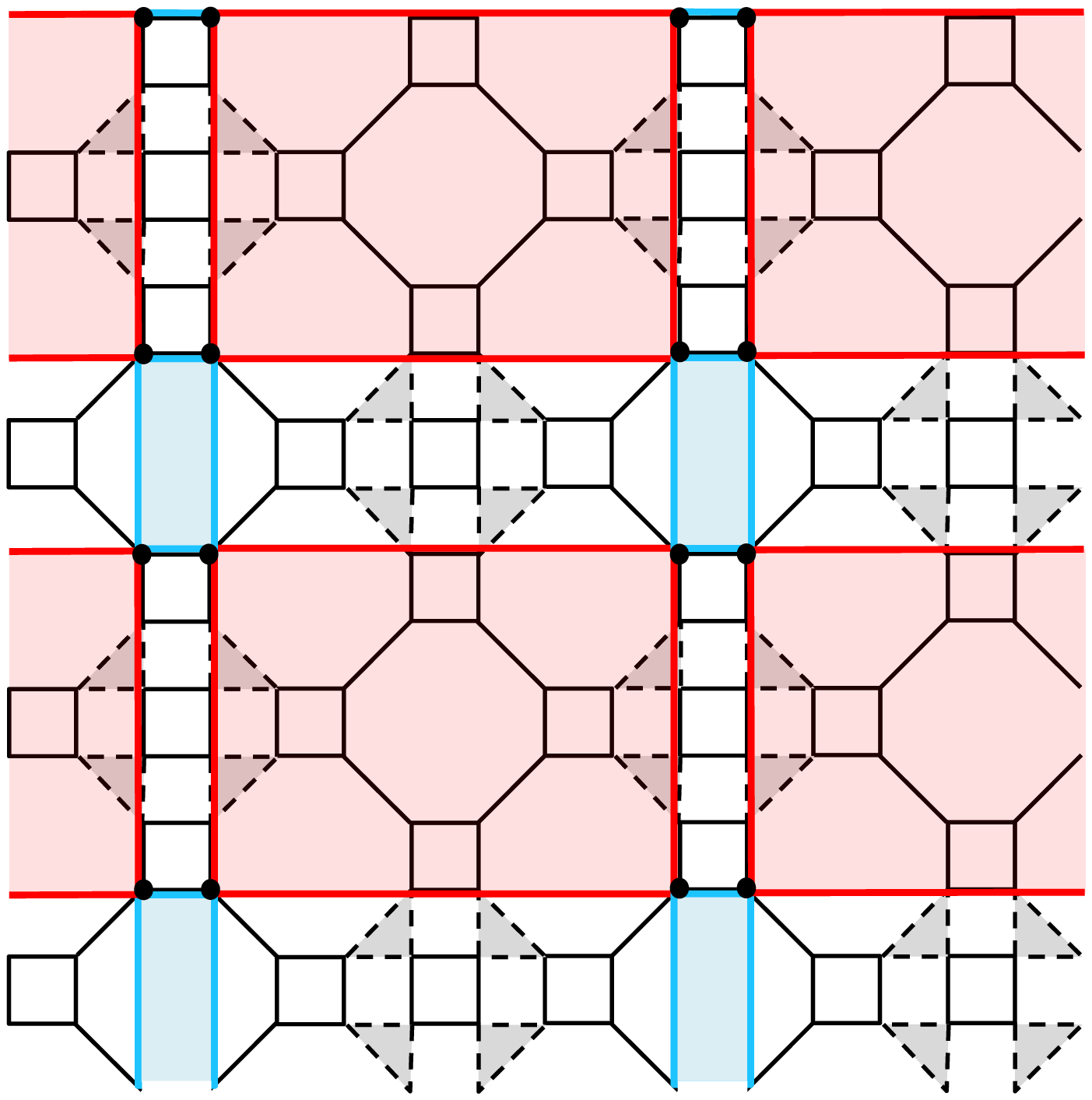}
\caption{(Color Online) The stabilizers $A$ and $C$ on a lattice with the four $Z$ errors in the unit cell.}
\label{fig:transformedLattice1}
\end{figure}

\begin{figure}[h!]
\centering
\includegraphics[width=.49\textwidth]{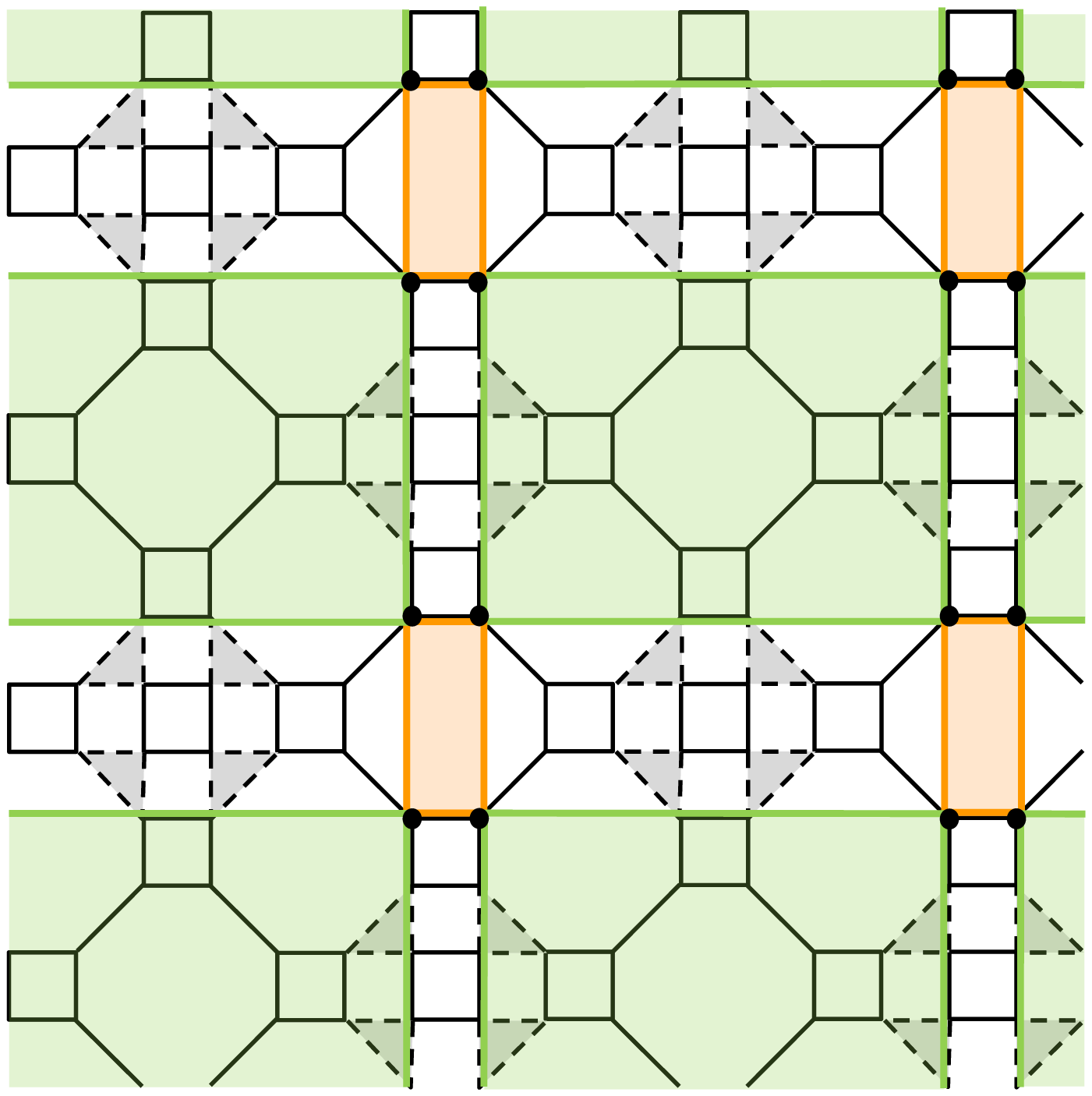}
\caption{(Color Online) The stabilizers $A$ and $C$ on a lattice with the four $Z$ errors in the unit cell.}
\label{fig:transformedLattice2}
\end{figure}

We can now try to give a rough estimate of the noiseless error correction threshold, solely based on this mapping onto decoding for the toric code. An upper-bound for the threshold against, say, $X$ errors for the toric code is $11\%$ \cite{WHP:threshold} and this is close that what minimum weight decoding achieves. We can derive a rough estimate for the effective single-qubit noise rate for $Z$ errors on the reduced toric code lattice, as a function of the basic depolarizing noise rate. On the original lattice, we have a single-qubit $Z$-error rate of $2p/3$ where $p$ is the depolarizing error rate. Some $X$ errors get corrected, but some also get mapped onto $Z$ errors so, perhaps one should consider that after the first round of error correction (the processing of the syndrome of operator $D$), one has $Z$ error rates between $2p/3$ and $p$ on all sites of the lattices (for example, the $Z$ error rate on qubit at the location "1" is $2p/3$). Using multiplication by the gauge group elements, we replace these $20$ possible error locations in a unit cell by $4$ effective error locations. Hence the effective $Z$-error on these $4$ error locations is between $10 p/3$ and $5p$. This would suggest that the five-squares code can achieve at most a depolarizing noise threshold between $2.2\%$ and $3.3\%$ using this decoding method.

This analysis is rough, but it suggests why this decoding strategy for the fives-squares code will lead to a threshold which is markedly lower than the threshold of the surface code. The topological subsystem code has a larger overhead in terms of number of qubits in order to make for very local gauge group generators. If this larger overhead does not contribute to the error correction, it simply increases the effective error rate that is seen by the toric code decoder.

In this rough estimate we have also neglected correlations between errors. As we mentioned earlier, elementary single qubit errors can be mapped onto two-qubit errors by multiplication with gauge group elements. Hence the effective error model that is seen by the toric code decoder is not an i.i.d. single-qubit error model. Our experimental determination of the threshold simulation of course takes these correlations into account. We believe that correlations do not play a major role in depressing the threshold.

% Figure - error transformation
%\begin{figure}[t!]
%\centering
%\includegraphics[width=.49\textwidth]{xErrorTransformation}
%\caption{Three equivalent errors. Gauge operators that demonstrate the equivalence are highlighted.}
%\label{fig:xErrorTransformation}
%\end{figure}

\subsection{Improved Decoding for the Five-Squares Code}
\label{Improved_ST_decoding}
% Main idea - smaller weight decoding
The decoding algorithm in the previous subsection is suboptimal because it does not minimize the weight of the corrected error. We reduce the weight of the corrected error by modifying our algorithm -- we correct the $X$ errors in the first step of the algorithm at such locations so that we reduce the number of $Z$ errors that result.

% Observation - X errors can be corrected at any single corner of the square
First, we observe that the $X$ errors do not necessarily need to be corrected at qubit $1$ as indicated in Figure~\ref{fig:xErrors}. Instead, if the $D$ syndrome corresponding to a particular square is nontrivial, we can correct the concerned bit flip at any single corner of the square. Further, we observe that depending on the location where we perform the correction, the $B$ syndrome in the corresponding unit cell will or will not change its sign. This is illustrated in Figure~\ref{fig:xErrorsImproved}.

% The main correction step
To describe the main correction step, let's focus on a particular unit cell and its syndromes $B$ and $D$. Let's first assume that at least one of the five $D$ syndromes is non-trivial, and the $B$ syndrome is trivial. Then, if we correct the bit flips indicated by the $D$ syndromes at the locations that do not affect the $B$ syndrome, we know that the $Z$ error will not occur at qubit $4$. If there is a least one non-trivial $D$ syndrome, and the $B$ syndrome is non-trivial as well, we perform exactly one of the $X$ error corrections at a location that changes the $B$ syndrome. Again, the $Z$ error will not occur at qubit $4$. The last case to analyze has all $D$ syndromes trivial, and the $B$ syndrome non-trivial. In such a case the $Z$ error correction at qubit $4$ still needs to be performed, and the improved algorithm does not reduce the weight of the corrected error. In the first two analyzed cases, the weight of the corrected error decreases by up to one.

% Figure - location of the errors
\begin{figure}[t]
\centering
\subfigure[$X$ error corrections at these locations do not affect the $B$ syndrome.]{
\includegraphics[width=.22\textwidth]{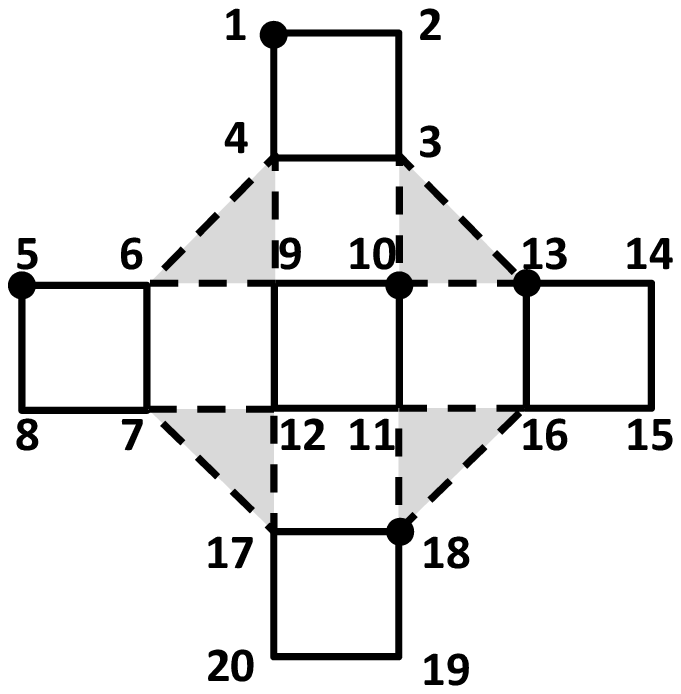}
\label{fig:xErrorsImproved1}
}
\hspace{0.6mm}
\centering
\subfigure[$X$ error corrections at these locations change the sign of the $B$ syndrome.]{
\includegraphics[width=.22\textwidth]{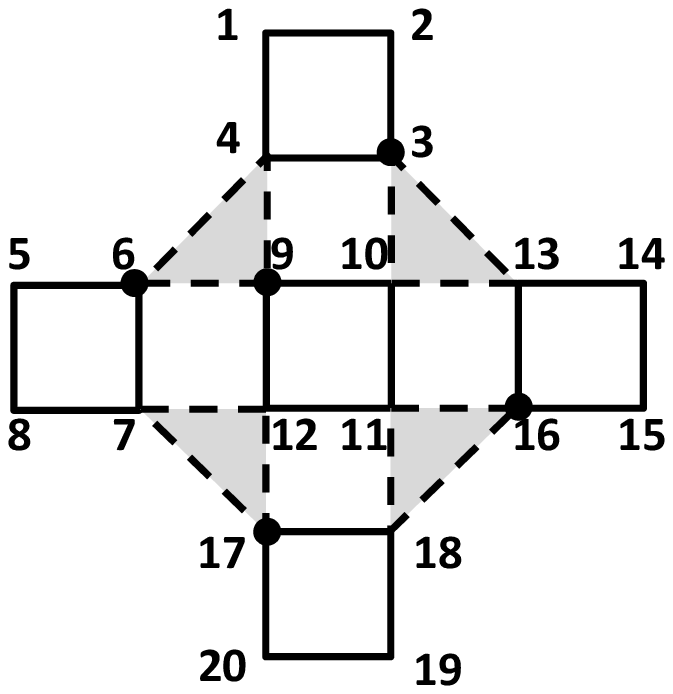}
\label{fig:xErrorsImproved2}
}
\caption{Illustration of the improved decoding.}
\label{fig:xErrorsImproved}
\end{figure}

\subsection{Syndrome Measurements in the Five-Squares Code}
% Stabilizer measurements
An important aspect concerning the practicality of the five-squares code, in particular the noise threshold against noisy error correction, concerns the number of two-qubit gauge element measurement needed to determine the syndrome of the stabilizer generators. The probability with which a stabilizer generator has a faulty syndrome is experimentally determined by the error rate on the two-qubit gauge element measurements.

Consider, for example, stabilizer $B$ in the five-squares code. Figure~\ref{fig:stabilizersMeasurement} shows that $10$ link operators need to be measured to determine the syndrome. It is easy to verify that to determine syndromes $A$, $C$, and $D$ in the five-squares code, respectively, $8$, $40$, and $4$ two-qubit measurements need to be taken. One may then anticipate that the probability of a faulty stabilizer generator syndromes, say for a stabilizer of type $C$ is roughly at most $40$ times the probability for a faulty two-qubit gauge element measurement. This implies that the noise threshold against noisy error correction may be considerably lower than the $2\%$ noise rate that we obtain here. We note that for the square-octagon code the number of measurements to determine the four types of syndromes is slightly smaller, namely $14$, $8$, $4$, and $20$.

% Figure - stabilizer measurements
\begin{figure}[t!]
\centering
\includegraphics[width=.45\textwidth]{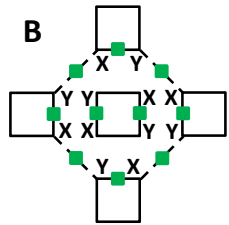}
\caption{Example -- these ten link operators need to be measured to determine syndrome $B$ of the five-squares code.}
\label{fig:stabilizersMeasurement}
\end{figure}

\section{Experimental Evaluation of the Threshold}
\label{Experiments}

% Overview
We performed extensive simulations to evaluate the threshold properties of the simple and improved decoding algorithm. We experimented with a range of lattice sizes and depolarizing error rates, recording the ratio of unsuccessful error corrections for each such configuration. This allowed us to accurately estimate the threshold because for large lattice sizes, the decoding succeeds with high probability for error rates below the threshold, and fails otherwise.

% The algorithm that we used
Our simulation first generates a random depolarizing error $E$, and then it utilizes the simple or improved decoding algorithm to determine the error $E_{guessed}$ that should be corrected according to the decoder. In the final step, we determine whether correcting the original error $E$ by applying the guessed error $E_{guessed}$ fails, i.e., if $E E_{guessed}$ is proportional to a logical operator. More precisely, for a particular depolarizing error rate $p$, and a fixed lattice size, the simulation proceeds as follows:
\begin{enumerate}
    \item Generate a random error $E$ with depolarizing probability $p$ in the five-squares lattice.
    \item Calculate syndromes $A$, $B$, $C$, and $D$.
    \item Add $X$ errors indicated by syndromes $D$ to $E_{guessed}$, and update syndromes $A$, $B$, and $C$ assuming that the error $E_{guessed}$ is corrected (i.e., the Pauli frame is updated).
    \item Add $Z$ errors indicated by syndromes $B$ to $E_{guessed}$, and update syndromes $A$, and $C$ accordingly.
    \item Correct the errors in the two sub-lattices in Figure~\ref{fig:transformedLattice1} and \ref{fig:transformedLattice2} separately. Use the minimum weight matching algorithm, and add $Z$ errors on the shortest lines connecting the matched non-trivial syndromes to $E_{guessed}$, as described in detail below.
    \item Calculate $O = EE_{guessed}$. If $O$ is a logical operator on the five-squares lattice, correction fails. Otherwise it succeeds. To determine if $O$ is a logical operator, we only need to check if it anticommutes with one of the logical operators $\bar{X}_i$ or $\bar{Z}_i$.
    \item The simulation above is repeated for a different random errors $E$.
\end{enumerate}

% How does simple and improved decoding differ
Note that the simulations of simple and improved decoding differ in how the third step is executed. In particular, in the third step, simple decoding guesses an $X$ error in the north-west corner of each square with non-trivial syndrome $D$, whereas improved decoding uses also syndrome $B$ to guess the specific location of the $X$ error within the square with non-trivial $D$ syndrome, as described in Section~\ref{Improved_ST_decoding}. Note that both decoders correct $Z$ errors at location $4$ in the unit cell (see Figure \ref{fig:yErrors}) in the fourth step if the corresponding syndrome $B$ is non-trivial.

% Minimum weight matching - how is it calculated
The minimum weight matching in the fifth step of the algorithm is used as follows. Since the fifth step is repeated twice, once using syndromes depicted in Figure~\ref{fig:transformedLattice1}, and then analogously using syndromes depicted in Figure~\ref{fig:transformedLattice2}, we only describe the former. We identify the syndromes $A$ and $C$ in Figure~\ref{fig:transformedLattice1} that are non-trivial, and represent each such syndrome as a vertex. Then we calculate the pairwise distance of all such vertices and use the minimum weight matching algorithm to pair up vertices so that the sum of the distances between the vertex pairs is minimized. The distance of syndrome $u$ and $v$ in the graph represents the minimal number of $Z$ errors that have to be introduced to the original lattice in Figure~\ref{fig:transformedLattice1} to change syndromes $u$ and $v$ while keeping all other syndromes intact. The $Z$ errors that need to be introduced to change syndromes $u$ and $v$ lie on the shortest line connecting the two syndromes. These $Z$ errors are added to $E_{guessed}$ in the fifth step of the simulation.

% Parameters
We executed the experiment with the depolarizing error model with error probability ranging from $0\%$ to $5\%$ with $0.1\%$ increments. The experiment was repeated for square lattices consisting of a total of $640$, $2,560$, $10,240$, and $40,960$ qubits. Each of these square lattices consisted of $n$ by $2n$ unit cells where each unit cell has $20$ qubits. For example, the lattice in Figure~\ref{fig:transformedLattice1} has $2$ by $4$ unit cells for a total of $160$ qubits. For each considered depolarizing error rate and lattice size, we repeated the error correction $1,000$ times and recorded the number of experiments that resulted in a failed decoding.

% Interpretation of the results
The results are depicted in Figure~\ref{fig:thresholdSimple} for the simple decoding algorithm, and Figure~\ref{fig:thresholdJoint} for the improved decoding algorithm. We observe that in both cases the code exhibits a threshold. In the former case the threshold is around $1.5\%$, and in the later case around $2\%$.

% The figures show what we would expect
As the error rate increases, the failure probability in Figure~\ref{fig:thresholdSimple} and~\ref{fig:thresholdJoint} approaches $15/16$, which is as expected since the success probability for picking the correct coset at random is 1/16.

% Why the square-octagon experiments are not done here
We did not perform experiments with the square-octagon code as this code and its noise threshold with noiseless error correction are studied in \cite{PBGC:square_oct}. Surprisingly the authors find a similar depolarizing noise threshold of $2\%$.

% Figure - simple decoding
\begin{figure}[t!]
\centering
\includegraphics[width=.49\textwidth]{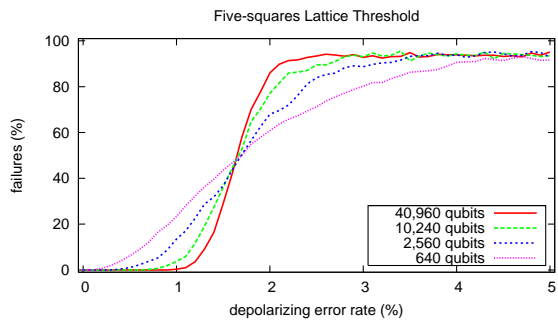}
\caption{The threshold is around 1.5\%.}
\label{fig:thresholdSimple}
\end{figure}

% Figure - joint B and D syndrome correction
\begin{figure}[t!]
\centering
\includegraphics[width=.49\textwidth]{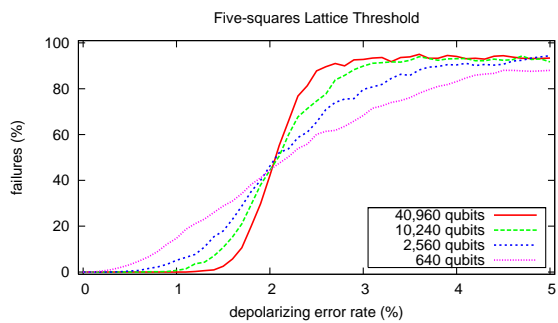}
\caption{The threshold is around 2\%.}
\label{fig:thresholdJoint}
\end{figure}

\section{Conclusion}
\label{Conclusion}

We studied the error correction properties of topological subsystem codes whose gauge group has local two-qubit generators. We provided an example of such code on a two-dimensional grid of qubits and we described a candidate decoding algorithm. Our Monte Carlo simulations indicated that our code exhibits a threshold for a depolarizing error rate of around 2\%. This is substantially lower than the threshold for the surface code with noise-free error correction. We have tried to assess whether the threshold achieved by our improved decoder is close to the theoretical, achievable, threshold in Eq.~(\ref{eq:belowabove}). For this we have numerically executed a simulated annealing algorithm that performs minimum-weight decoding. The algorithm solves the problem of finding the minimum energy of the 2D local Hamiltonian with quenched disorder determined by $\bar{P}E$ by gradually lowering the temperature. It then compares the found minima for various $\bar{P}$ and picks the one with the lowest minimum. Our numerics indicate that this decoder does not improve on the decoders described in the previous section. It is an important open question whether there are other topological subsystem codes which have a threshold which is closer to that of the surface code or whether the lower threshold is an inevitable price that we pay by enhancing the locality, 2-body versus 4-body, of the code.

\ack
The authors would like to thank Hector Bombin and David Poulin for discussions concerning topological subsystem codes and sharing their unpublished work~\cite{PBGC:square_oct} in this area. The authors acknowledge fruitful and motivating discussions with David DiVincenzo concerning the use of topological subsystem codes.  BMT and SB acknowledge support by the DARPA QUEST program under contract number HR0011-09-C-0047. MS would like to thank the Physics of Information Group at IBM Research where this project originated for support, mentoring, and fruitful discussions.

\section*{Appendix~A}
In this section we prove that  an eigenvalue of a  stabilizer $S$ can be determined using eigenvalue measurements of the gauge group generators $K_j$ only if $S$ admits a decomposition Eq.~(\ref{syndrom0}) satisfying Eq.~(\ref{syndrom1}).

Indeed, the most general measurement protocol involves {\em adaptive} eigenvalue measurements: suppose we have already measured eigenvalues of some generators $K_1,\ldots,K_j$ obtaining eigenvalues $\gamma_1,\ldots,\gamma_j\in \{+1,-1\}$. The next generator to be measured, $K_{j+1}$, may be chosen as an arbitrary function of the outcomes $\gamma_1,\ldots,\gamma_j$ and, may be,  some extra random bits. Let
\[
\Pi_j =\frac12 (I + \gamma_j K_j)
\]
be the projector describing the eigenvalue measurement of $K_j$.
Let
\[
R_j=\Pi_j \cdots \Pi_1
\]
be the operator describing the entire sequence of measurements up to the step $j$  for some fixed configuration of outcomes $\gamma=(\gamma_1,\ldots,\gamma_m)$. We shall say that the above adaptive measurement protocol  simulates eigenvalue measurement of $S$ if there exist a function $\sigma\, : \, \{+1,-1\}^m \to \{+1,-1\}$ such that
\be
\label{simulation1}
SR_m=R_m S = \sigma(\gamma)\, R_m
\ee
for all $\gamma\in \{+1,-1\}^m$. Equivalently, the reduced state obtained after the last measurement must be an eigenvector of $S$ with an eigenvalue $\sigma(\gamma)$. Below we prove that this is possible only if $S$ admits a decomposition Eq.~(\ref{syndrom0}) satisfying Eq.~(\ref{syndrom1}), may be, with a different $m$.

Indeed, for any fixed outcomes $\gamma$ define a family of stabilizer codes (Abelian groups of Pauli operators) $\calS_0,\calS_1,\ldots,\calS_m\subseteq \calP_n$ such that $\calS_0=\la I\ra$, $\calS_1=\la \gamma_1 K_1\ra$, and
\be
\label{part_stabilizer1}
\calS_{j+1}=\la \gamma_{j+1} K_{j+1}, \; \calS_j'\ra,
\ee
where
\be
\label{part_stabilizer2}
\calS_j'=\{ P\in \calS_j \, : \, PK_{j+1} = K_{j+1} P\}
\ee
for $j=1,\ldots,m-1$ and $\calS_m'=\calS_m$. Using the standard stabilizer formalism~\cite{book_nielsen_chuang} one can easily check  that $R_j = \Pi_j U_j$, where $\Pi_j$ is the projector onto the codespace of the stabilizer code $\calS_j$ and $U_j$ is some Clifford group unitary operator that belongs to the algebra generated by $K_1,\ldots,K_j$. In particular, both $\Pi_j$ and $U_j$ commute with $S$. Obviously, condition Eq.~(\ref{simulation1}) holds only if $\sigma(\gamma) S \in \calS_m$. Using the definitions Eqs.~(\ref{part_stabilizer1},\ref{part_stabilizer2}) we conclude that
\be
\label{part_stabilizer3}
\sigma(\gamma) S = (\gamma_m K_m)^{\alpha(m)} S_{m-1} \quad \mbox{for some} \quad S_{m-1} \in \calS_{m-1}'
\ee
and some $\alpha(m)\in \{0,1\}$. Note that $K_m$ commutes with $S_{m-1}$ by definition of the group  $\calS_{m-1}'$. Using the inclusion $\calS_j'\subseteq \calS_j$ inductively we arrive at the decomposition
\be
\label{part_stabilizer4}
\sigma(\gamma) S= (\gamma_m K_m)^{\alpha(m)}  \cdots (\gamma_1 K_1)^{\alpha(1)}
\ee
for some $\alpha(1),\ldots,\alpha(m)\in \{0,1\}$. This decomposition obeys $(\gamma_j K_j)^{\alpha(j)} \cdots (\gamma_1 K_1)^{\alpha(1)} \in \calS_j'$ for all $j=1,\ldots,m$. Hence $K_j^{\alpha(j)}$ commutes with $K_{j-1}^{\alpha(j-1)} \cdots K_1^{\alpha(1)}$ for all $j$. Omitting all the factors with $\alpha(j)=0$ we arrive at the decomposition Eqs.~(\ref{syndrom0},\ref{syndrom1}).

% References
\section*{References}

%\bibliographystyle{unsrt}
%\bibliography{RefCodes}

\begin{thebibliography}{10}

\bibitem{scheme_reducing_decoherence}
P.~W. Shor.
\newblock Scheme for reducing decoherence in quantum computer memory.
\newblock {\em Phys. Rev. A}, 52(4):R2493, Oct 1995.

\bibitem{simple_quantum_error}
A.~M. Steane.
\newblock Simple quantum error-correcting codes.
\newblock {\em Phys. Rev. A}, 54(6):4741, Dec 1996.

\bibitem{mixed_state_entanglement}
Charles~H. Bennett, David~P. DiVincenzo, John~A. Smolin, and William~K.
  Wootters.
\newblock Mixed-state entanglement and quantum error correction.
\newblock {\em Phys. Rev. A}, 54(5):3824, Nov 1996.

\bibitem{theory_quantum_error}
Emanuel Knill and Raymond Laflamme.
\newblock Theory of quantum error-correcting codes.
\newblock {\em Phys. Rev. A}, 55(2):900, Feb 1997.

\bibitem{class_quantum_error}
Daniel Gottesman.
\newblock Class of quantum error-correcting codes saturating the quantum
  hamming bound.
\newblock {\em Phys. Rev. A}, 54(3):1862--1868, Sep 1996.

\bibitem{fault_tolerant_quantum}
A.~Yu. Kitaev.
\newblock Fault-tolerant quantum computation by anyons.
\newblock {\em Annals of Physics}, 303(1):2, Jan 2003.

\bibitem{topological_quantum_memory}
Eric Dennis, Alexei Kitaev, Andrew Landahl, and John Preskill.
\newblock Topological quantum memory.
\newblock {\em J. Math. Phys.}, 43(9):4452, Sep 2002.

\bibitem{BK:surface}
S.~B. {Bravyi} and A.~Y. {Kitaev}.
\newblock {Quantum codes on a lattice with boundary}.
\newblock {\em ArXiv quant-ph/9811052}, 1998.

\bibitem{bombin:topsub}
H.~{Bombin}.
\newblock {Topological subsystem codes}.
\newblock {\em Phys.~Rev.~A}, 81(3):032301, March 2010.

\bibitem{BMD:codedef}
H.~{Bombin} and M.~A. {Martin-Delgado}.
\newblock {Quantum Measurements and Gates by Code Deformation}.
\newblock {\em ArXiv 0704.2540}, 2007.

\bibitem{RH:cluster2D}
R.~{Raussendorf} and J.~{Harrington}.
\newblock {Fault-Tolerant Quantum Computation with High Threshold in Two
  Dimensions}.
\newblock {\em Phys. Rev. Lett.}, 98(19):190504, May 2007.

\bibitem{bombin:codedef_topsub}
H.~{Bombin}.
\newblock {Clifford Gates by Code Deformation}.
\newblock {\em ArXiv 1006.5260}, 2010.

\bibitem{BK:magicdist}
S.~{Bravyi} and A.~{Kitaev}.
\newblock {Universal quantum computation with ideal Clifford gates and noisy
  ancillas}.
\newblock {\em Phys.~Rev.~A}, 71(2):022316, February 2005.

\bibitem{CDT:codestudy}
A.~W. {Cross}, D.~P. {DiVincenzo}, and B.~M. {Terhal}.
\newblock {A comparative code study for quantum fault-tolerance}.
\newblock {\em ArXiv 0711.1556}, 2007.

\bibitem{RHG:threshold}
R.~{Raussendorf}, J.~{Harrington}, and K.~{Goyal}.
\newblock {Topological fault-tolerance in cluster state quantum computation}.
\newblock {\em New J. Phys.}, 9:199--219, June 2007.

\bibitem{wangetal:threshold}
D.~S. {Wang}, A.~G. {Fowler}, A.~M. {Stephens}, and L.~C.~L. {Hollenberg}.
\newblock {Threshold error rates for the toric and surface codes}.
\newblock {\em ArXiv 0905.0531}, 2009.

\bibitem{WFH:highthresh}
D.~S. {Wang}, A.~G. {Fowler}, and L.~C.~L. {Hollenberg}.
\newblock {Quantum computing with nearest neighbor interactions and error rates
  over 1\%}.
\newblock {\em ArXiv 1009.3686}, 2010.

\bibitem{DCP:topdec_prl}
G.~{Duclos-Cianci} and D.~{Poulin}.
\newblock {Fast Decoders for Topological Quantum Codes}.
\newblock {\em Phys. Rev. Lett.}, 104(5):050504, February 2010.

\bibitem{DCP:topdec_long}
Guillaume Duclos-Cianci and David Poulin.
\newblock {A renormalization group decoding algorithm for topological quantum
  codes}.
\newblock {\em ArXiv 1006.1362}, 2010.

\bibitem{WHP:threshold}
C.~{Wang}, J.~{Harrington}, and J.~{Preskill}.
\newblock {Confinement-Higgs transition in a disordered gauge theory and the
  accuracy threshold for quantum memory}.
\newblock {\em Annals of Physics}, 2003.

\bibitem{divincenzo:JJarchitecture}
D.~P. {DiVincenzo}.
\newblock {Fault-tolerant architectures for superconducting qubits}.
\newblock {\em Physica Scripta Volume T}, 137(1):014020, December 2009.

\bibitem{subsystem_codes_spatially}
Sergey Bravyi.
\newblock {Subsystem codes with spatially local generators}.
\newblock {\em ArXiv 1008.1029}, 2010.

\bibitem{stabilizer_formalism_operator}
David Poulin.
\newblock Stabilizer formalism for operator quantum error correction.
\newblock {\em Phys. Rev. Lett.}, 95(23):230504, Dec 2005.

\bibitem{operator_quantum_error}
Dave Bacon.
\newblock Operator quantum error-correcting subsystems for self-correcting
  quantum memories.
\newblock {\em Phys. Rev. A}, 73(1):012340, Jan 2006.

\bibitem{AC:baconshor}
P.~{Aliferis} and A.~W. {Cross}.
\newblock {Subsystem Fault Tolerance with the Bacon-Shor Code}.
\newblock {\em Phys. Rev. Lett.}, 98(22):220502, June 2007.

\bibitem{BHM:topo}
S.~{Bravyi}, M.~B. {Hastings}, and S.~{Michalakis}.
\newblock {Topological quantum order: Stability under local perturbations}.
\newblock {\em J. Math. Phys.}, 51(9):093512, September 2010.

\bibitem{Kitaev05}
A.~{Kitaev}.
\newblock {Anyons in an exactly solved model and beyond}.
\newblock {\em Ann. Phys.}, 321(1):2, 2006.

\bibitem{yao_kivelson}
H.~{Yao} and S.~A. {Kivelson}.
\newblock {Exact Chiral Spin Liquid with Non-Abelian Anyons}.
\newblock {\em Phys. Rev. Lett.}, 99(24):247203, December 2007.

\bibitem{BT08}
S.~Bravyi and B.~M. Terhal.
\newblock {A no-go theorem for a two-dimensional self-correcting quantum memory
  based on stabilizer codes}.
\newblock {\em New. J. Phys.}, 11:043029, 2009.

\bibitem{chinese_postman_problem}
Jack Edmonds.
\newblock The {C}hinese postman problem.
\newblock {\em Operations Research}, 13:373, 1965.

\bibitem{PBGC:square_oct}
H.~Bombin, G.~Duclos-Cianci, and D.~Poulin, 2010.
\newblock To appear.

\bibitem{BMD:topo}
H.~Bombin and M.~A. Martin-Delgado.
\newblock Topological quantum distillation.
\newblock {\em Phys. Rev. Lett.}, 97, 2006.

\bibitem{book_nielsen_chuang}
Michael~A. Nielsen and Isaac~L. Chuang.
\newblock {\em Quantum Computation and Quantum Information}.
\newblock Cambridge University Press, 2000.

\end{thebibliography}

\end{document}